\pdfoutput=1
\documentclass[sigconf,authorversion]{acmart}

\usepackage{algorithm}
\usepackage[noend]{algpseudocode}
\newcommand{\Ifline}[2]{\State \textbf{if }#1{ \textbf{then} }#2} 
\algrenewcomment[1]{\hfill \(\triangleright\) \emph{\small #1}} 
\usepackage{amsmath}
\usepackage{bbm}
\usepackage{mathtools}
\usepackage{xspace}
\usepackage[utf8]{inputenc}

\usepackage[capitalise]{cleveref}
\usepackage{mathtools}

\newcommand{\ZZ}{\ensuremath{\mathbb{Z}}}
\newcommand{\ZZpos}{\ensuremath{\mathbb{Z}_{>0}}}

\newcommand{\field}{\ensuremath{\mathbb{K}}}
\newcommand{\ext}{\ensuremath{\mathbb{L}}}
\newcommand{\FF}[1]{\ensuremath{\mathbb{F}_{#1}}}
\renewcommand{\vec}[1]{\boldsymbol{#1}}
\newcommand{\xdeg}{\deg_x}
\newcommand{\ydeg}{\deg_y}
\newcommand{\T}{\mathcal T}
\newcommand{\M}{\mathcal{M}}
\newcommand{\FxIso}{\phi}

\newcommand{\set}[1]{\{{#1}\}}

\DeclareMathOperator{\rank}{rank}
\DeclareMathOperator{\rem}{rem}
\DeclareMathOperator{\Fail}{``Fail''}
\DeclareMathOperator{\cdeg}{cdeg}
\DeclareMathOperator{\degdet}{\Delta}

\newcommand{\Grobner}{Gr\"obner\xspace}
\newcommand{\lex}{{\mathrm{lex}}}
\newcommand{\ordLex}{\prec_\lex} 
\newcommand{\LTlex}{\mathsf{LT}_\lex}
\newcommand{\LCy}{\mathsf{LC}_y}

\newcommand{\pts}{\mathcal P}
\newcommand{\vals}{\vec \gamma}

\newcommand{\valX}{{\nu_x}}
\newcommand{\valY}{{\nu_y}}
\newcommand{\VanIdeal}{\Gamma}
\newcommand{\degSeq}{\vec\eta}

\DeclarePairedDelimiter\floor{\lfloor}{\rfloor}
\DeclarePairedDelimiter\ideal{\langle}{\rangle}
\newcommand{\degTop}{\eta}
\newcommand{\gPols}{\vec g}
\newcommand{\val}{\gamma}
\newcommand{\bigO}[1]{\mathchoice{O\left(#1\right)}{O(#1)}{O(#1)}{O(#1)}} 
\newcommand{\softO}[1]{\mathchoice{\tilde{O}\left(#1\right)}{\tilde{O}(#1)}{\tilde{O}(#1)}{\tilde{O}(#1)}} 
\newcommand{\expmm}{\omega} 
\newcommand{\Mult}{\mathsf{M}} 
\newcommand{\algoname}[1]{{\normalfont\textsc{#1}}} 
\newcommand{\assign}{\leftarrow} 
\newcommand{\Input}{\Statex \textbf{Input: }} 
\newcommand{\Output}{\Statex \textbf{Output: }} 
\newcommand{\Preinput}{\Statex \textbf{Preinput: }} 
\newcommand{\Precomputation}{\Statex \textbf{Precomputation: }} 
\newcommand{\Precompute}{\algoname{ComputeReshaper}}

\newcommand{\quorem}{\algoname{quo\_rem}}
\newcommand{\Reshape}{\algoname{Reshape}}
\newcommand{\Interpolate}{\algoname{Interpolate}}
\newcommand{\BivariateMPE}{\algoname{MPE-DistinctX}}
\newcommand{\ValencyMPE}{\algoname{MPE-Shear}}

\newcommand{\Shift}{\algoname{ShearPoly}}
\newcommand{\ModComp}{\algoname{ModComp}}

\newcounter{algsubstate}
\renewcommand{\thealgsubstate}{\alph{algsubstate}}
\newenvironment{algsubstates}
  {\setcounter{algsubstate}{0}%
   \renewcommand{\State}{%
     \stepcounter{algsubstate}%
     \Statex {\footnotesize\thealgsubstate:}\hspace{10pt}}}
  {}

\copyrightyear{2020}
\acmYear{2020}
\setcopyright{acmlicensed}\acmConference[ISSAC '20]{International Symposium on Symbolic and Algebraic Computation}{July 20--23, 2020}{Athens, Greece}
\acmBooktitle{International Symposium on Symbolic and Algebraic Computation (ISSAC '20), July 20--23, 2020, Athens, Greece}
\acmPrice{15.00}
\acmDOI{10.1145/3373207.3404032}
\acmISBN{978-1-4503-7100-1/20/07}

\begin{document}

\title{Generic Bivariate Multi-point Evaluation, Interpolation and Modular Composition with Precomputation}

\author{Vincent Neiger}
\affiliation{%
  \institution{{\normalsize Univ. Limoges, CNRS, XLIM, UMR\,7252}}
  \city{F-87000 Limoges}
  \state{France}
}

\author{Johan Rosenkilde}
\affiliation{%
  \institution{{\normalsize Technical University of Denmark}}
  \city{Kgs.~Lyngby}
  \state{Denmark}
}

\author{Grigory Solomatov}
\affiliation{%
  \institution{{\normalsize Technical University of Denmark}}
  \city{Kgs.~Lyngby}
  \state{Denmark}
}

\begin{abstract}
  Suppose $\field$ is a large enough field and $\pts \subset \field^2$ is a fixed, generic set of points which is available for precomputation.
  We introduce a technique called \emph{reshaping} which allows us to design quasi-linear algorithms for both:
  computing the evaluations of an input polynomial $f \in \field[x,y]$ at all points of $\pts$;
  and computing an interpolant $f \in \field[x,y]$ which takes prescribed values on $\pts$ and satisfies an input $y$-degree bound.
  Our genericity assumption is explicit and we prove that it holds for most point sets over a large enough field.
  If $\pts$ violates the assumption, our algorithms still work and the performance degrades smoothly according to a distance from being generic.
  To show that the reshaping technique may have an impact on other related problems, we apply it to modular composition: suppose generic polynomials $M \in \field[x]$ and $A \in \field[x]$ are available for precomputation, then given an input $f \in \field[x,y]$ we show how to compute $f(x, A(x)) \rem M(x)$ in quasi-linear time.
\end{abstract}

\keywords{Multi-point evaluation, interpolation, modular composition, bivariate polynomials, precomputation.}

\maketitle

\section{Introduction}
\label{sec:intro}

\paragraph{Outline}
Let $\field$ be an effective field.
We consider the three classical problems for bivariate polynomials $\field[x,y]$ mentioned in the title.
We assume a model where part of the input is given early as \emph{preinput} which is available for heavier computation, and the primary goal is to keep the complexity of the \emph{online phase}, once the remaining part of the input is given, to a minimum.

\textbf{Multi-point evaluation (MPE):} with preinput a point set $\pts = \{(\alpha_i,\beta_i)\}_{i=1}^n \subseteq \field^2$ and input $f \in \field[x,y]$, compute $\big(f(\alpha_i,\beta_i)\big)_{i=1}^n$.
We give two algorithms: the first requires pairwise distinct \(\alpha_i\)'s and has online complexity $\softO{\xdeg f \ydeg f + n}$ as long as $\pts$ is \emph{balanced}, a notion described below;
  the second accepts repeated $x$-coordinates with online complexity $\softO{\xdeg f(\xdeg f + \ydeg f) + n}$ as long as a certain ``shearing'' of $\pts$ is balanced.
``\emph{soft-O}'' ignores logarithmic terms: $\bigO{f(n)(\log f(n))^c} \subset \softO{f(n)}$ for any $c \in \ZZ_{\geq 0}$.

\textbf{Interpolation:} with preinput a point set $\pts$ as before, and input values $\vec\gamma \in \field^n$, compute $f \in \field[x,y]$ such that $\big(f(\alpha_i,\beta_i)\big)_{i=1}^n = \vec\gamma$, satisfying some constraints on the monomial support.
We give an algorithm which preinputs a degree bound $d$ and outputs $f$ such that $\ydeg f < d$ and $\xdeg f \in \bigO{n/d}$.
The online complexity is $\softO{n}$ if $\pts$ and a shearing of $\pts$ are both balanced; $d$ should exceed the \emph{$x$-valency of $\pts$}, i.e.~the maximal number of $y$-coordinates for any given $x$-coordinate.

\textbf{Modular composition:} with preinput $M, A \in \field[x]$, we input $f \in \field[x,y]$ and compute $f(x, A) \rem M$.
Our algorithm has online complexity $\softO{\xdeg f \ydeg f + \deg A + \deg M}$, as long as the bivariate ideal $\ideal{M, y-A}$ is balanced.

We prove that if $\pts \subseteq \field^2$ is random of fixed cardinality $n$, and if $|\field| \gg n^2 \log(n)$ then $\pts$ is balanced with high probability.
Similarly, if $M$ is square-free and $A$ is uniformly random of degree less than \(\deg M\), then $\ideal{M, y-A}$ is balanced with high probability.
Our proof techniques currently do not extend to proving that sheared point sets are balanced.
A few trials we conducted suggest that this may often be the case if the $x$-valency of $\pts$ is not too high.
The cost of the second MPE algorithm is not symmetric in the $x$- and $y$-degree, so whenever
$\xdeg f < \ydeg f$ one should consider transposing the input, i.e.~evaluating $f(y, x)$ on $\{ (\beta_i, \alpha_i) \}_{i=1}^n$.
In this case, the balancedness assumption is on the transposed point set.

Our algorithms are deterministic, and once the preinput has been processed, the user knows whether it is balanced and hence whether the algorithms will perform well.
Further, the performance of our algorithms deteriorates smoothly with how ``unbalanced'' the preinput is, in the sense of certain polynomials, which depend only on preinput, having sufficiently well behaved degrees.
In a toolbox one might therefore apply our algorithms whenever the preinput turns out to be sufficiently balanced and reverting to other algorithms on very unbalanced preinput.

A typical use of precomputation is if we compute e.g.~MPEs on the same point set for many different polynomials.
This occurs in coding theory, where bivariate MPE corresponds to the encoding stage of certain families of codes such as some Reed-Muller codes \cite[Chap.\,5]{assmus_designs_1992} and some algebraic-geometric codes \cite{miura_algebraic_1993}: here $\pts$ is fixed and communication consists of a long series of bivariate MPEs on $\pts$.
In these applications, $\pts$ is often not random but chosen carefully, and so our genericity assumptions might not apply.

\paragraph{Techniques}

We introduce a tool we call \emph{reshaping} for achieving the following: given an ideal $I \subseteq \field[x,y]$ and $f \in \field[x,y]$, compute $\hat f \in f + I$ with smaller $y$-degree.
For instance in MPE, we let $\VanIdeal \subset \field[x,y]$ be the ideal of polynomials which vanish on all the points $\pts$.
Then all elements of $f + \VanIdeal$ have the same evaluations on $\pts$, so we compute a $\hat f \in f + \VanIdeal$ of $y$-degree $0$ (it exists if $\pts$ has distinct $x$-coordinates), and then apply fast univariate MPE.

An obvious idea to accomplish this iteratively is to find some $g \in \VanIdeal$ of lower $y$-degree than $f$ and whose leading $y$-term is $1$, and then compute $\tilde f = f \rem g$.
The problem is that the $x$-degree of $\tilde f$ may now be as large as $\xdeg f + (\ydeg f - \ydeg g)\xdeg g$.
Our idea is to seek polynomials $g$ that we call \emph{reshapers}, which have the form
\begingroup
\abovedisplayskip=2pt
\belowdisplayskip=2pt
\[
  g = y^{2d/3} - \hat{g} \ ,
\]
\endgroup
where $\ydeg \hat{g} < d/3$ and $d = \ydeg f+1$ (for simplicity, here \(3\) divides \(d\)).
Writing $f = f_{1} y^{2d/3} + f_{0}$ with $\ydeg f_{0} < 2d/3$,
then $\tilde f = f_{1} \hat{g} + f_{0}$ is easy to compute, has $y$-degree less than $2d/3$, and $x$-degree only $\xdeg f + \xdeg g$.
Repeating such a reduction $\bigO{\log(d)}$ times with reshapers of progressively smaller $y$-degree, we eventually reach $y$-degree $0$.

For efficiency, we therefore need the $x$-degrees of all these reshapers $g$ to be small.
For MPE, stating that $g \in \VanIdeal$ specifies $n$ linear contraints on the coefficients of $\hat{g}$, so we look for \(g\) with about $n$ monomials.
Generically, since $\ydeg \hat{g} \approx d/3$, one may expect to find \(g\) with $\xdeg g \approx 3n/d$.
Informally, $\pts$ is \emph{balanced} if all the reshapers needed in the above process satisfy this degree constraint.

Above, we assumed the point set has distinct $x$-coordinates.
To handle repetitions, we shear the points by $(\alpha, \beta) \mapsto (\alpha + \theta \beta, \beta)$, where $\theta$ generates an extension field of \(\field\) of degree 2.
The resulting point set has distinct $x$-coordinates.
This replaces $f(x,y)$ with $f(x - \theta y, y)$, and whenever $\xdeg f < \ydeg f$ we stay within quasi-linear complexity if the sheared point set is balanced.

\paragraph{Previous work}
Quasi-linear complexity has been achieved for multivariate MPE and interpolation on special point sets and monomial support: Pan \cite{pan_simple_1994} gave an algorithm on grids, and van der Hoeven and Schost \cite{van_der_hoeven_multi-point_2013} (see also \cite[Sec.\,2]{coxon_fast_2018}) generalised this to certain types of subsets of grids, constraining both the points and the monomial support.
See \cite{van_der_hoeven_multi-point_2013} for references to earlier work on interpolation, not achieving quasi-linear complexity.

In classical univariate modular composition, we are given $f, M, A$ in $\field[x]$ and seek $f(A) \rem M$.
Brent and Kung's baby-step giant-step algorithm \cite{paterson_number_1973,brent-kung-1978} performs this operation in $\softO{n^{(\omega+1)/2}}$, where $\omega$ is the matrix multiplication exponent with best known bound $\omega < 2.373$ \cite{le_gall_powers_2014}.
N\"usken and Ziegler \cite{nusken-ziegler-2004} extended this to a bivariate $f$, computing $f(x, A) \rem M$ in complexity $\bigO{\xdeg f(\ydeg f)^{(\omega+1)/2}}$, assuming that $A$ and $M$ have degree at most $\xdeg f \ydeg f$.
They applied this to solve MPE in the same cost; in this paper,
we use essentially the same link between these problems.
To the best of our knowledge, this is currently the best known cost bound for these problems,
in the algebraic complexity model.

In a breakthough, Kedlaya and Umans \cite{kedlaya_fast_2011} achieved
``almost linear'' time for modular composition and MPE, for specific types of
fields \(\field\) and in the bit complexity model.
For modular composition, the cost is \(\bigO{n^{1+\epsilon}}\) bit operations for any \(\epsilon>0\),
while for MPE it is $\bigO{(n + (\xdeg f)^2)^{1 + \epsilon}}$,
assuming $\ydeg f < \xdeg f$ (the algorithm also supports multivariate MPE).
Unfortunately, these algorithms have so far resisted attempts at a practical implementation \cite{van_der_hoeven_fast_2019}.

Our quasi-linear complexities improve upon the above results (including Kedlaya and Umans' ones since quasi-linear compares favorably to almost linear);
however we stress that none of the latter have the two constraints of our work: allowing precomputation, and genericity assumption.
For modular composition, precomputation on $M$ was suggested in \cite{van_der_hoeven_modular_2018} to leverage its factorisation structure.
Except for slight benefits of precomputation in Brent and Kung's modular composition (used in the Flint and NTL libraries \cite{flint,ShoupNTL}),
we are unaware of previous work focusing on the use of precomputation for MPE, Interpolation, and Modular Composition.

Genericity has recently been used by Villard \cite{villard_computing_2018}, who showed how to efficiently compute the resultant of two generic bivariate polynomials; a specific case computes, for given univariate \(M\) and \(A\), the characteristic polynomial of \(A\) in \(\field[x] / \ideal M\), with direct links to the modular composition \(f(A) \rem M\) \cite{villard_computing_2018,villard_computing_2018_slides}.
This led to an ongoing work on achieving exponent \((\omega+2)/3\) for modular composition \cite{NeigerSalvySchostVillard}.
In that line, the main benefit from genericity is that $\ideal{M, y - A}$ admits bases formed by \(m\) polynomials of \(y\)-degree \(<m\) and \(x\)-degree at most \(\deg(M)/m\), for a given parameter \(2 \le m \le \deg(M)\).
Such a basis is represented as an \(m \times m\) matrix over \(\field[x]\) with all entries of degree at most \(\deg(M)/m\), and one can then rely on fast univariate polynomial matrix algorithms.
In this paper, genericity serves a purpose similar to that in \cite{villard_computing_2018,NeigerSalvySchostVillard}: it ensures the existence of such bases for several parameters \(m\), and also of the reshapers \(g\) mentioned above; besides we make use of these bases to precompute these reshapers.
Whereas an important contribution of \cite{villard_computing_2018} is the efficient computation of such bases, here they are only used to find reshapers in the precomputation stage and the speed of computing them is not a main concern.
Once the reshapers are known, our algorithms work without requiring any other genericity property.

\paragraph{Organisation} After some preliminaries in \cref{sec:preliminaries}, we describe the reshaping strategy for an arbitrary ideal in \cref{sec:reshape}.
Then \cref{sec:mpe,sec:interpolation,sec:modcomp} give algorithms for each of the three problems.
We discuss precomputation in \cref{sec:precomputing} and genericity in \cref{sec:genericity}.

\section{Preliminaries}
\label{sec:preliminaries}

For complexity estimates, we use the algebraic RAM model and count arithmetic operations in $\field$.
By $\Mult(n)$ we denote the cost of multiplying two univariate polynomials over $\field$ of degree at most $n$;
one may take $\Mult(n) \in \bigO{n \log n \log\log n} \subset \softO{n}$ \cite{cantor_fast_1991}.
Division with remainder in \(\field[x]\) also costs $\bigO{\Mult(n)}$ \cite[Thm.\,9.6]{von_zur_gathen_modern_2012}.
When degrees of a polynomial, say $f \in \field[x,y]$, appear in complexity estimates, we abuse notation and let $\xdeg f$ denote $\max(\xdeg f, 1)$.

It is well-known that univariate interpolation and multi-point evaluation can be done in quasi-linear time \cite[Cor.\,10.8 and\,10.12]{von_zur_gathen_modern_2012}:
given $f \in \field[x]$ and $\alpha_1,\ldots,\alpha_n \in \field$, we may compute $\big(f(\alpha_i)\big)_{i=1}^n$ in time $\bigO{\Mult(\xdeg f + n)\log n} \subseteq \softO{\xdeg f + n}$;
given $\alpha_1,\ldots,\alpha_n$ and $\beta_1,\ldots,\beta_n$ in $\field$ with the \(\alpha_i\)'s pairwise distinct, we may compute the unique corresponding interpolant in time $\bigO{\Mult(n)\log n} \subseteq \softO{n}$.
We will also use the fact that two bivariate $f, g \in \field[x,y]$ can be multiplied in time $\bigO{\Mult(d_x d_y)} \subset \softO{d_x d_y}$, where $d_x = \max(\xdeg f, \xdeg g)$ and $d_y = \max(\ydeg f, \ydeg g)$ \cite[Cor.\,8.28]{von_zur_gathen_modern_2012}.

For a bivariate polynomial $f = \sum_{i=0}^k f_i(x) y^i \in \field[x,y]$ such that \(f_k \neq 0\), we define its $y$-leading coefficient as $\LCy(f) = f_k \in \field[x]$.

For our genericity results, we will invoke the following staple:
\begin{lemma}[DeMillo-Lipton-Schwartz-Zippel \protect{\cite{Schwartz80,Zippel79,DeMilloLipton78}}]
  \label{lem:zippel}
  Let $f \in \field[x_1,\ldots,x_n]$ be non-zero of total degree \(d\), and $\T \subseteq \field$ be finite.
  For $\alpha_1,\ldots,\alpha_k \in \T$ chosen independently and uniformly at random,
  the probability that $f(\alpha_1,\ldots,\alpha_k) = 0$ is at most $d/|\T|$.
\end{lemma}

For a point set $\pts \subseteq \field^2$, the $x$-valency of $\pts$, denoted by $\valX(\pts)$, is the largest number of $y$-coordinates for any given $x$-coordinate, i.e.
\[
  \valX(\pts) = \max_{\alpha \in \field}|\{ \beta \in \field \mid (\alpha,\beta) \in \pts \}| \ .
\]
When $\valX(\pts) = 1$, the $x$-coordinates of $\pts$ are pairwise distinct.

The vanishing ideal of \(\pts\) is the bivariate ideal
\[
  \VanIdeal(\pts) = \{ f \in \field[x,y] \mid f(\alpha,\beta) = 0 \text{ for all } (\alpha,\beta) \in \pts \} \ ,
\]
Hereafter, \(\ordLex\) stands for the lexicographic order on \(\field[x,y]\) with \(x \ordLex y\), and $\LTlex(f)$ is the \(\ordLex\)-leading term of \(f\in\field[x,y]\).
The following is folklore and follows e.g.~from \cite{lazard_ideal_1985} and \cite[Thm.\,3]{Dahan2009}.

\begin{lemma}
  \label{lem:gb_vanideal}
  Let \(\pts \subset \field^2\) be a point set of cardinality $n$ and let $G = \{g_1,\ldots,g_s\}$ be the reduced $\ordLex$-\Grobner basis of $\VanIdeal(\pts)$, ordered by $\ordLex$.
  Then $g_1 \in \field[x]$, and $g_s$ is $y$-monic with $\ydeg g_s = \valX(\pts)$.
\end{lemma}

\section{Reshape}
\label{sec:reshape}

We first describe our algorithm \Reshape{} which takes $f \in \field[x,y]$ and an ideal $I$ and finds $\hat f \in f + I$ whose $y$-degree is below some target.
This will pass through several intermediate elements of $f + I$ of progressively smaller $y$-degree.
This sequence of $y$-degrees has the following form:

\begin{definition}
  \label{def:reshaping_seq}
  We say $\degSeq = (\degTop_i)_{i=0}^k \in \ZZ_{> 0}^{k+1}$ is a $(\degTop_0, \degTop_k)$-\emph{reshaping sequence} if
  $\degTop_{i-1} > \degTop_{i} \geq \floor{ \tfrac 2 3 \degTop_{i-1}}$ for $i=1,\ldots,k$.
  For $I \subseteq \field[x,y]$ an ideal and $\degSeq = ( \degTop_i )_{i=0}^k$ a reshaping sequence,
  we say $\gPols = ( g_i )_{i=1}^k \in I^k$ is an \emph{$\degSeq$-reshaper for $I$} if $g_i = y^{\degTop_i} + \hat g_i$ where $\ydeg \hat g_i \leq 2\degTop_i - \degTop_{i-1}$, for each $i = 1,\ldots,k$.
\end{definition}

Our algorithms are faster with short reshaping sequences, so we should choose $\degTop_i \approx \frac 2 3 \degTop_{i-1}$, and hence $2\degTop_i - \degTop_{i-1} \approx \frac 1 3 \degTop_i$.
It is easy to see that for any $a,b \in \ZZpos$, there is an $(a,b)$-reshaping sequence of length less than $\log_{3/2}(a) + 2$.
Observe that for any $(a,b)$-reshaping sequence we have $\eta_i \geq \frac 2 3 (\eta_{i-1}-1)$ for $i=1,\ldots,k$ and therefore
\begin{equation}
  \label{eqn:eta_diff}
  \textstyle
  2\eta_i - \eta_{i-1} \geq \frac {\eta_{i-1} - 4} 3 \geq \frac {\eta_i} 3 - 1 \ .
\end{equation}
By considering the cases $\eta_{i} \geq 3$ and $\eta_i = 1,2$, we get \(2\eta_i - \eta_{i-1} \geq 0\).

\begin{theorem} \label{thm:Reshape}
  \cref{algo:Reshape} is correct and has complexity
  \begin{align*}
    & \textstyle\softO{\sum_{i=i_0}^k \degTop_i (\xdeg f + \sum_{j=i_0}^i  \xdeg g_j)} \\ 
    & \subseteq \textstyle\softO{k \ydeg f\xdeg f + k\sum_{i=i_0}^k \degTop_i \xdeg g_i} \ ,
  \end{align*}
  for the smallest \(i_0\) such that \(\degTop_{i_0} \le \ydeg f\).
\end{theorem}

\begin{proof}
  Let $\hat f_i, \hat f_{i,0}, \hat f_{i,1}$ be the values of $\hat f, \hat f_0, \hat f_1$ at the end of iteration $i$.
  First, the iterations for \(i<i_0\) perform no operation and keep \(\hat f_i = f\), since \(\degTop_i > \ydeg \hat f_{i-1}\) implies \(\hat f_{i,1} = 0\) and \(\hat f_i = \hat f_{i-1}\).
  In particular, if \(\degTop_i>\ydeg f\) for all \(i\) then the algorithm is correct and returns \(f\) without using any arithmetic operation.
  Now for \(i\ge i_0\), observe that
  \(
    \hat{f}_{i} = \hat{f}_{i,1}\hat{g}_i + \hat f_{i,0}
                = \hat{f}_{i-1}-\hat{f}_{i,1} g_i;
  \)
  thus in the end $\hat f \in f + I$ since each $g_i$ belongs to $I$.
  We show the following loop invariants, which imply the degree bounds on the output:
  {\par\centering
     $\xdeg \hat f_i \leq \xdeg f + \sum_{j=i_0}^i \xdeg g_j$,
     and $\ydeg \hat f_i < \degTop_i$.\par
  }
  \noindent Both are true for $i=i_0-1$ (just before the loop, if \(i_0=1\)).
  For the $x$-degree, $\hat f_i = \hat f_{i-1} - \hat f_{i,1} g_i$ yields $\xdeg \hat f_i \leq \xdeg \hat f_{i-1} + \xdeg g_i$, and the loop invariant follows.
  For the $y$-degree, by construction $\ydeg \hat f_{i,0} < \degTop_i$ and $\ydeg \hat{f}_{i,1} \leq \ydeg \hat{f}_{i-1} - \degTop_i$ hold; the assumption $\ydeg \hat{f}_{i-1} < \degTop_{i-1}$ then gives $\ydeg\hat f_{i,1} \hat g_i < \degTop_i$, hence $\ydeg \hat f_i < \degTop_i$.

  For complexity, the only costly step is at \cref{Reshape:reduce} and for iterations \(i\ge i_0\).
  From the above bound $\ydeg \hat f_{i,1} \hat g_i < \degTop_i$, multiplying $\hat f_{i,1}$ and $\hat g_i$ costs $\bigO{\Mult((\xdeg \hat f_{i,1} + \xdeg \hat g_i) \degTop_i)}$.
  Since $\xdeg \hat g_i = \xdeg g_i$, since both \(\hat f_{i,0}\) and \(\hat f_{i,1}\) have \(x\)-degree at most \(\xdeg \hat f_{i-1}\), and since $\ydeg \hat f_{i,0} < \degTop_i$, the total cost of the \(i\)th iteration is in
  \[
    \textstyle
    \softO{(\xdeg \hat f_{i-1} + \xdeg \hat g_i)\degTop_i}
      \subseteq \softO{(\xdeg f + \sum_{j=i_0}^i \xdeg g_j)\degTop_i}.
  \]
  Summing over all iterations, we get the first complexity bound in the theorem; the second one follows from it, using the fact that $\ydeg f \ge \degTop_{i_0} > \degTop_{i_0+1} > \ldots > \degTop_k$ and \(i_0\ge1\).
\end{proof}

\begin{algorithm}[t]
  \caption{\Reshape$(f,\degSeq,\gPols)$} \label{algo:Reshape}
  \begin{algorithmic}[1]
    \Input A bivariate polynomial $f \in \field[x,y]$;
          a reshaping sequence $\degSeq = (\degTop_i)_{i=0}^k \in \ZZpos^{k+1}$ with $\ydeg f < \degTop_0$;
          an $\degSeq$-reshaper $\gPols = (g_i)_{i=1}^k \in I^k$ for some ideal $I \subseteq \field[x,y]$.
    \Output a polynomial $\hat{f} \in f + I$ such that $\ydeg \hat{f} < \degTop_k$ and
    $\xdeg \hat{f} \leq \xdeg f + \sum_{i = 1}^{k}\xdeg g_i$.
    \State $\hat{f} \assign f$
    \For{$i = 1,\dots,k$}
    \State Write $g_i = y^{\degTop_i} + \hat g_i$ where $\ydeg \hat g_i \leq 2\degTop_i - \degTop_{i-1}$
    \State Write $\hat{f} = \hat{f}_{1}y^{\degTop_i} + \hat{f}_{0}$ where $\ydeg \hat{f}_{0} < \degTop_i$ \label{Reshape:split}
    \State $\hat{f} \assign \hat{f}_{1}\hat{g}_i + \hat f_{0}$ \Comment{equivalent to $\hat{f} \assign \hat{f} - \hat{f}_{1}g_i$} \label{Reshape:reduce}
    \EndFor
    \State \Return $\hat{f}$
  \end{algorithmic}
\end{algorithm}

We now define the balancedness of a point set. In \cref{sec:genericity} we prove that this notion captures the \emph{expected} $x$-degree of reshapers.

\begin{definition} \label{def:balanced}
  Let $\pts \subseteq \field^2$ be a point set of cardinality $n$, and let $\degSeq = ( \degTop_i )_{i=0}^k$ be a reshaping sequence.
  Then $\pts$ is \emph{$\degSeq$-balanced} if there exists an $\degSeq$-reshaper $\gPols = (g_i)_{i=1}^k \in \field[x,y]^k$ for $\VanIdeal(\pts)$ such that
  $\xdeg g_i \leq \floor{\tfrac {n} {2\eta_i - \eta_{i-1} + 1}} + 1$ for $i=1,\ldots,k$.
\end{definition}

The next bound is often used below for deriving complexity estimates; it follows directly from \cref{eqn:eta_diff}.

\begin{lemma} \label{lem:balanced}
  Let $\degSeq = (\degTop_i)_{i=0}^k$ be a reshaping sequence, $\pts \subseteq \field^2$ be an $\degSeq$-balanced point set of cardinality $n$, and $\gPols = (g_i)_{i=1}^k$ be an $\degSeq$-reshaper for $\VanIdeal(\pts)$.
  Then $\sum_{i=i_0}^k \degTop_i \xdeg g_i \le (3n+\degTop_{i_0})k$ for \(1 \le i_0 \le k\).
\end{lemma}

We conclude this section with two results about the existence of $\degSeq$-reshapers for vanishing ideals of point sets.

\begin{lemma} \label{lem:exist_reshapers}
  Let $\pts \subseteq \field^2$ be a point set and $\degSeq = ( \degTop_i )_{i=0}^k$ a reshaping sequence.
  If $\valX(\pts) \le \min_{1\le i\le k}(2\degTop_i - \degTop_{i-1} + 1)$, then there exists an $\degSeq$-reshaper $\gPols \in \field[x,y]^k$ for $\VanIdeal(\pts)$.
\end{lemma}
\begin{proof}
  By \cref{lem:gb_vanideal}, the reduced \(\ordLex\)-Gr\"obner basis \(G\) of \(\VanIdeal(\pts)\) contains a polynomial with \(\ordLex\)-leading term $y^{\valX(\pts)}$.
  Thus $\ydeg y^\degTop \rem G < \valX(\pts)$ for any $\degTop$, and setting $g_i = y^{\degTop_i} - (y^{\degTop_i} \rem G)$ yields an $\degSeq$-reshaper as long as $\valX(\pts) \leq 2\degTop_i - \degTop_{i-1} + 1$ for all \(i\).
\end{proof}

\begin{corollary}
  \label{cor:reshapers_constrained}
  Let $\pts \subseteq \field^2$ be a point set of cardinality $n$ and $a, b \in \ZZpos$ with $n > a > b \geq \valX(\pts)$.
  Then there is an $(a,b)$-reshaping sequence $\degSeq$ which satisfies the condition of \cref{lem:exist_reshapers} and has length $k \leq \log_{3/2}(a)+1 \in O(\log(a))$.
\end{corollary}
\begin{proof}
  Let $v = \valX(\pts)-1$ and let $\degSeq' = (\degTop'_{0},\dots,\degTop'_{k})$ be any $(a-v,b-v)$-reshaping sequence with $k \leq \log_{3/2}(a-v)+1$. Now let $\degSeq = (\degTop_0,\dots,\degTop_k)$ be defined by $\degTop_{i} = \degTop'_{i}+v$ for $i = 0,\dots,k$. Then, $\degSeq$ is an $(a,b)$-reshaping sequence. Indeed, clearly the endpoints are correct and $\degTop_{i-1} > \degTop_{i}$ for $i = 1,\ldots,k$; moreover,
  \[
    \degTop_i = \degTop'_{i}+v \geq \floor{\tfrac{2}{3}\degTop'_{i-1}} + v = \floor{\tfrac{2}{3}\degTop_{i-1} + \tfrac13 v} \geq \floor{\tfrac{2}{3}\degTop_{i-1}} \ .
  \]
  To conclude, we use \(2\degTop'_i - \degTop'_{i-1} \ge 0\) as mentioned above to observe that
  \(
    2 \degTop_i - \degTop_{i-1} + 1
    = 2\degTop'_i - \degTop'_{i-1} + v + 1 \ge v+1 = \valX(\pts)
  \).
\end{proof}

\section{Multi-Point Evaluation}
\label{sec:mpe}

In this section we use reshaping for MPE with precomputation; i.e. given a point set $\pts \subset \field^2$ upon which we are allowed to perform precomputation, and a polynomial $f \in \field[x,y]$ which is assumed to be received at online time, compute $f(P)$ for all $P \in \pts$.
\cref{algo:BivariateMPE} deals with the case $\valX(\pts) = 1$, which we reduce to an instance of univariate MPE using \Reshape.
The cost of \cref{algo:BivariateMPE} follows directly from \cref{thm:Reshape,lem:balanced}.

\begin{algorithm}[h]
    \caption{$\BivariateMPE_{d,\degSeq,\pts}(f)$} \label{algo:BivariateMPE}
  \begin{algorithmic}[1]
    \Preinput $d \in \ZZpos$; a $(d,1)$-reshaping sequence $\degSeq$;
    a point set $\pts = \{ (\alpha_i,\beta_i) \}_{i=1}^n \subset \field^2$ with the $\alpha_i$'s pairwise distinct.
    \Precomputation
    \begin{algsubstates}
        \State $\gPols \assign$ $\degSeq$-reshaper for $\VanIdeal(\pts)$
    \end{algsubstates}
    \Input $f \in \field[x,y]$ with $\ydeg f < d$.
    \Output $\big(f(\alpha_1,\beta_1), \dots, f(\alpha_n,\beta_n) \big) \in \field^n$.
    \State $\hat{f} \assign \Reshape(f, \degSeq,\gPols) \in \field[x]$
    \State \Return $\big(\hat{f}(\alpha_1), \dots, \hat{f}(\alpha_n) \big) \in \field^n$ \Comment{univariate MPE}
    \label{line:BivariateMPE:uniMPE}
  \end{algorithmic}
\end{algorithm}
\begin{theorem}
  \label{thm:BivariateMPE}
  \cref{algo:BivariateMPE} is correct.
  If $\pts$ is $\degSeq$-balanced and \(\degSeq\) has length in $\bigO{\log(n)}$, the complexity is $\softO{\xdeg f \ydeg f + n}$.
\end{theorem}

\cref{algo:BivariateMPE} can easily be extended to the case where $\valX(\pts) > 1$ by partitioning $\pts$ into $\valX(\pts)$ many subsets, each having $x$-valency one. This approach also has quasi-linear complexity in the input size as long as $\valX(\pts) \ll n$, or more precisely if $n\valX(\pts) \in \softO{n}$.

When $\valX(\pts)$ is large, this strategy is costly, and we proceed instead by shearing the point set, as proposed by N\"usken and Ziegler \cite{nusken-ziegler-2004}, so that the resulting point set has distinct $x$-coordinates:
by taking $\theta \in \ext \setminus \field$, where $\ext$ is an extension field of $\field$ of degree 2, we apply the map $(\alpha,\beta) \mapsto (\alpha + \theta \beta,\beta)$ to each element of $\pts$.
The problem then reduces to evaluating $\bar{f} = f(x - \theta y, y)$ at the sheared points.
To compute $\bar f$, \cite{nusken-ziegler-2004} provides an algorithm with complexity $\bigO{\Mult(d_x(d_x + d_y))\log(d_x)}$ using a univariate Taylor shift of \(f\) seen as a polynomial in \(x\) over the ring \(\ext[y]\).
\cref{algo:Shift} describes an algorithm for this task which improves the cost on the logarithmic level, by using Taylor shifts of the homogeneous components of \(f\).

{
\def\xshear{a}
\def\yshear{b}
\begin{algorithm}[H]
  \caption{\Shift$(f, \xshear , \yshear )$}
  \label{algo:Shift}
  \begin{algorithmic}[1]
    \Input $f = \sum_{i=0}^{d_x} \sum_{j=0}^{d_y}f_{i,j}x^iy^j \in \ext[x,y]$;
          $\xshear \in \ext$ and $\yshear  \in \ext$.
    \Output $f(\xshear  x + \yshear  y,y)$.
    \For{ $t = 0, \dots, d_x + d_y$ }
    \State $h_t \assign \sum_{i = \max(0,t - d_y)}^{\min(t,d_x)} f_{i,t-i}z^{i} \in \ext[z]$
    \State $s_t \assign h_t(\xshear z + \yshear )$ \Comment{Taylor shift}
    \label{Shift:uni-shift}
    \EndFor
    \State \Return $\sum_{t = 0}^{d_x + d_y} y^t s_t(x/y)$
  \end{algorithmic}
\end{algorithm}
\begin{theorem} \label{thm:Shift}
  \cref{algo:Shift} correctly computes $f(\xshear  x + \yshear  y,y)$, which has \(x\)-degree at most $d_x$ and \(y\)-degree at most $d_x + d_y$,
  at a cost of
  \(
    \bigO{(d_x + d_y)\Mult(d_x)\log(d_x)} \subset \softO{d_x(d_x + d_y)}
  \)
  operations in $\ext$.
\end{theorem}
\begin{proof}
  Observe that \(y^th_t(x/y)\) is the homogeneous component of \(f\) of degree \(t\), and in particular $f = \sum_{t = 0}^{d_x + d_y} y^th_t(x/y)$.
  Thus
  \[
  f(\xshear  x + \yshear  y,y)
  = \textstyle\sum_{t=0}^{d_x + d_y}y^t h_t\left(\frac{\xshear  x + \yshear  y}{y}\right)
  = \textstyle\sum_{t=0}^{d_x + d_y}y^t s_t(x/y),
  \]
  hence the correctness. The degree bounds on the output are straightforward.
  As for complexity, only \cref{Shift:uni-shift} uses arithmetic operations.
  First, scaling \(h_t(z) \mapsto h_t(\xshear  z)\) costs \(\bigO{d_x}\) operations in \(\ext\), since \(\deg h_t \le d_x\);
  then the Taylor shift \(h_t(\xshear  z) \mapsto h_t(\xshear  z + \yshear )\) costs $\bigO{\Mult(d_x)\log(d_x)}$ operations in \(\ext\) according to \cite[Fact 2.1(iv)]{Gathen1990}.
  Summing over the \(d_x+d_y\) iterations yields the claimed bound.
\end{proof}
}

This leads to \cref{algo:ValencyMPE}, where \(\pts\) may have repeated \(\alpha_i\)'s.

\begin{algorithm}[ht]
  \caption{$\ValencyMPE_{d,\degSeq,\pts}(f)$} \label{algo:ValencyMPE}
  \begin{algorithmic}[1]
    \Preinput an integer $d \in \ZZpos$; a $(d,1)$-reshaping sequence $\degSeq$;
    a point set $\pts = \{ (\alpha_i,\beta_i) \}_{i=1}^n \subset \field^2$.
    \Precomputation
    \begin{algsubstates}
      \State $(\ext,\theta) \assign$ degree $2$ extension of $\field$, element $\theta \in \ext \setminus \field$
      \State $\bar{\pts} \assign \set{(\alpha_i + \theta \beta_i, \beta_i)}_{i = 1}^n \subset \ext^2$
      \State Do the precomputation of $\BivariateMPE_{d,\degSeq,\bar{\pts}}$
    \end{algsubstates}
    \Input $f \in \field[x,y]$ with $\xdeg f + \ydeg f < d$.
    \Output $\big(f(\alpha_1,\beta_1), \dots, f(\alpha_n,\beta_n) \big) \in \field^n$.
    \State $\bar{f} \assign \Shift(f,1,-\theta )$ \Comment{$\bar f = f(x - \theta y,y)$}
    \State \Return $\BivariateMPE_{d,\degSeq,\bar{\pts}}(\bar{f})$
  \end{algorithmic}
\end{algorithm}
\begin{theorem}
  \cref{algo:ValencyMPE} is correct.
  If $\bar \pts$ is $\degSeq$-balanced and $\degSeq$ has length in $\bigO{\log(n)}$, its complexity is $\softO{\xdeg f(\xdeg f + \ydeg f) + n}$.
\end{theorem}

\section{Interpolation}
\label{sec:interpolation}

In this section we use reshaping for the interpolation problem in a similar setting: we input a point set $\pts$ for precomputation, and input interpolation values at online time.
When $\pts$ is appropriately balanced, we solve the interpolation problem in quasi-linear time (see \cref{algo:interp}).
The strategy is to first shear the point set to have unique $y$-coordinates and compute $u \in \ext[y]$ which interpolates the values on the sheared $y$-coordinates.
Then we reshape this into $r \in \ext[x,y]$ with $x$- and $y$-degrees roughly $\sqrt n$.
Shearing back this polynomial to interpolate the original point set is now in quasi-linear time;
a last reshaping allows us to meet the target $y$-degree.

\begin{algorithm}[ht]
  \caption{$\Interpolate_{d,\degSeq,\pts}(\vals)$}
  \label{algo:interp}
  \begin{algorithmic}[1]
    \Preinput an integer $d \in \ZZpos$;
    an $(n, d)$-reshaping sequence $\degSeq = (\degTop_{i})_{i=0}^{k}$ such that $\degTop_{k_1} = \floor{\sqrt{n}}$ for some $k_1$;
    a point set $\pts = \{(\alpha_i,\beta_i)\}_{i=1}^n \subseteq \field^2$ such that $\valX(\pts) \leq d \leq \floor{\sqrt n} + 1$ and \(\valX(\pts) \leq \min_{1\le i\le k}(2\degTop_i - \degTop_{i-1} + 1)\).
    \Precomputation
    \begin{algsubstates}
      \State $\degSeq_1 \assign (\degTop_i)_{i=0}^{k_1}$ and $\degSeq_2 \assign (\degTop_i)_{i=k_1}^k$
      \State $(\ext,\theta) \assign
      \begin{cases}
        (\field,0) \text{ if } \valY(\pts) = 1 \\
        \text{degree } 2 \text{ extension of } \field, \theta \in \ext \setminus \field \text{ otherwise}
      \end{cases}$
      \State $\bar{\pts} \assign \set{(\alpha_i,\bar{\beta}_i)}_{i=1}^n$, where $\bar\beta_i = \theta \alpha_i + \beta_i$
      \State $\gPols_1 \assign$ $\degSeq_1$-reshaper for $\bar{\pts}$
      \State $\gPols_2 \assign$ $\degSeq_2$-reshaper for $\pts$
    \end{algsubstates}
    \Input Interpolation values $\vals = (\val_i)_{i=1}^n \in \field^n$.
    \Output $f \in \field[x,y]$ satisfying $f(\alpha_i,\beta_i) = \val_i$ for $i = 1,\dots,n$, $\ydeg f < d$ and $\xdeg f \leq \floor{\sqrt{n}} + \sum_{g \in \gPols_1\!\cup\gPols_2} \xdeg g$.
    \State $u \in \ext[y]$ with $\deg u < n$ and $u(\bar{\beta}_i) = \val_i$ for $i = 1,\dots,n$ \label{interp:uniinterp}
    \State $r \assign \Reshape(u,\degSeq_{1},\gPols_{1}) \in \ext[x,y]$ \label{interp:reshape}
    \State $s \assign r(x,\theta x + y)$ \Comment{using \Shift}
    \State Write $s = s_1 + \theta s_2$, where $s_1,s_2 \in \field[x,y]$
    \State \Return $\Reshape(s_1,\degSeq_{2},\gPols_{2}) \in \field[x,y]$
  \end{algorithmic}
\end{algorithm}

\begin{theorem}
  \cref{algo:interp} is correct and has complexity
  \[
    \textstyle
    \tilde{O}\bigg(
      k_1 n + k_2 \Big(\sqrt{n} + \sum\limits_{j = 1}^{k_1} \xdeg g_{1,j}\Big)^2 + \sum\limits_{\ell=1}^2 k_\ell \sum\limits_{j=1}^{k_\ell} \eta_{\ell,k} \xdeg g_{\ell,j}
    \bigg).
  \]
  If $\bar\pts$ is $\degSeq_1$-balanced and $\pts$ is $\degSeq_2$-balanced, and both $\degSeq_1$ and $\degSeq_2$ have length in $\bigO{\log n}$, then the complexity is $\softO{n}$.
\end{theorem}
\begin{proof}
  First note that a reshaping sequence of length $O(\log n)$ and satisfying the preinput constraints exists, due to \cref{cor:reshapers_constrained} and the assumption $d \geq \valX(\pts)$.
  For correctness, observe that all points in $\bar{\pts}$ have pairwise distinct $y$-coordinates, so computing $u$ makes sense.
  Viewing $u$ as an element of $\ext[x,y]$ with $\xdeg u = 0$, we have $u(\alpha_i,\bar{\beta}_i) = \val_i$.
  By \cref{thm:Reshape} then $r$ has the same evaluations and $\ydeg r < \floor{\sqrt{n}}$ and $\xdeg r \leq \sum_{i=1}^{k_1} \xdeg g_{1,i}$.

  Then, in both cases $\valY(\pts)=1$ and $\valY(\pts)>1$, we have
  \[
    \val_i = r(\alpha_i,\bar\beta_i) = s(\alpha_i,\beta_i) = s_1(\alpha_i,\beta_i) + \theta s_2(\alpha_i,\beta_i)
  \]
  for $i = 1,\dots,n$. Since $s_1,s_2 \in \field[x,y]$ and all $\val_i$'s are in $\field$, we get $s_2(\alpha_i,\beta_i) = 0$ and $s_1(\alpha_i,\beta_i) = \gamma_i$ for $i = 1,\dots,n$.
  We also then have that $\ydeg s_1 \leq \ydeg s < \floor{\sqrt{n}}$ and
  \[
    \textstyle \xdeg s_1 \leq \xdeg s \leq \ydeg r + \xdeg r \leq \floor{\sqrt{n}} + \sum_{j = 1}^{k_1} \xdeg g_{1,j} \ .
  \]
  Thus, by \cref{thm:Reshape} again, the output \(f\) is such that $f(\alpha_i,\beta_i) = \val_i$ for $i = 1, \dots, n$, and $\ydeg f < d$, and
  \[
    \textstyle \xdeg f \leq \floor{\sqrt{n}} + \sum_{j=1}^{k_1} \xdeg g_{1,j} + \sum_{j=1}^{k_2} \xdeg g_{2,j} \ .
  \]
  The complexity bound gathers the calls to \cref{algo:Reshape,algo:Shift},
  and the relaxed cost assuming balancedness is due to \cref{lem:balanced}.
\end{proof}

\section{Modular Composition}
\label{sec:modcomp}

We now turn to the following modular composition problem: given $M, A \in \field[x]$ with $n := \xdeg M > \xdeg A$, and $f \in \field[x,y]$, compute
\begin{equation}
  \label{eqn:mod_comp}
  f(x, A(x)) \rem M(x) \in \field[x] \ .
\end{equation}
We consider the variant of the problem where $M$ and $A$ are available for precomputation.
Computing \eqref{eqn:mod_comp} is tantamount to computing the unique element of $(f + I) \cap \field[x]$ of degree less than $n$,
for the ideal $I = \ideal{M, y-A} \subseteq \field[x,y]$.
One can thus see this as a reshaping task: given $f$ of some $y$-degree, reshape it to a polynomial of $y$-degree $0$ while keeping it fixed modulo $I$: this is formalised as \cref{algo:ModComp}.

\begin{algorithm}
  \caption{$\ModComp_{d, \degSeq, M, A}(f)$} \label{algo:ModComp}
  \begin{algorithmic}[1]
    \Preinput $d \in \ZZ_{> 0}$; a $(d, 1)$-reshaping sequence $\degSeq$;
            polynomials $M, A \in \field[x]$ with $n := \xdeg M > \xdeg A$.
    \Precomputation
    \begin{algsubstates}
      \State $\gPols \assign$ $\degSeq$-reshaper for $\ideal{M, y-A}$
    \end{algsubstates}
    \Input $f \in \field[x,y]$ with $\ydeg f < d$.
    \Output $f(x, A) \rem M \in \field[x]$.
    \State $\hat f \assign \Reshape(f,\degSeq,\gPols) \in \field[x]$
    \State \Return $\hat f \rem M$ \Comment{univariate division with remainder}
  \end{algorithmic}
\end{algorithm}

Like for point sets above, if $\degSeq = (\degTop_i)_{i=0}^k$ is a reshaping sequence, we say that $I = \ideal{M, y-A}$ is $\degSeq$-balanced if there exists an $\degSeq$-reshaper $\gPols = (g_i)_{i=1}^k$ for $I$ such that $\xdeg g_i \leq \floor{\frac n {2\degTop_i - \degTop_{i-1}+1}} + 1$.

\begin{theorem}
  \cref{algo:ModComp} is correct.
  If $\ideal{M, y-A}$ is $\degSeq$-balanced and $\degSeq$ has length in $\bigO{\log(n)}$, the complexity is $\softO{\xdeg f \ydeg f + n}$.
\end{theorem}

\section{Precomputing Reshapers}
\label{sec:precomputing}

\subsection{Reshapers for general ideals}

Here we describe \cref{algo:Precompute} for precomputing reshapers for any zero-dimensional ideal $I \subseteq \field[x,y]$, given a $\ordLex$-\Grobner basis of~$I$.
It operates through the $\field[x]$-module $I_\delta := \{ f \in I \mid \ydeg f < \delta\}$, so we first expound the relation between this and $I$ as a corollary of Lazard's structure theorem on bivariate \(\ordLex\)-\Grobner bases \cite{lazard_ideal_1985}.

\begin{corollary}
  \label{cor:module_ideal}
  Let $G = \{ b_0,\ldots,b_s \} \subset \field[x,y]$ be a minimal \(\ordLex\)-\Grobner basis defining an ideal $I = \ideal{G}$.
  For $\delta \in \ZZpos$, let $I_\delta = \{ f \in I \mid \ydeg f < \delta \}$, let
  $\hat s = \max\{ i \mid \ydeg b_i < \delta, \ 0 \leq i \leq s \}$, let $d_i = \ydeg b_i$ for $0 \le i \le \hat s$ and $d_{\hat s+1} = \delta$.
  Then $I_\delta$ is a $\field[x]$-submodule of \(\field[x,y]_{\ydeg < \delta}\) which is free of rank $\delta - d_0$ and admits the basis
  \(
    \{ y^j b_i  \mid  0 \le j < d_{i+1} - d_i, 0 \le i \le \hat s\}.
  \)
\end{corollary}
A proof is given in appendix.
We will use the following $\field[x]$-module isomorphism which converts between bivariate polynomials of bounded $y$-degree and vectors over $\field[x]$:
for any $\delta \in \ZZ_{>0}$,
\[
  \FxIso_\delta: f = \textstyle\sum_{j=0}^{\delta-1} f_j(x)y^j \in \field[x,y] \mapsto [f_0,\dots,f_{\delta-1}] \in \field[x]^{1 \times \delta} \ .
\]
If $I$ is zero-dimensional then in \cref{cor:module_ideal} we have $d_0 = 0$ and $I_\delta$ has rank $\delta$.
Any basis $B$ of $I_\delta$ can be represented as a nonsingular matrix $M_B \in \field[x]^{\delta \times \delta}$ whose rows are $\FxIso_\delta(B)$.
Then, $\degdet(I_\delta) := \deg\det(M_B)$ does not depend on the choice of $B$ since all bases of \(I_\delta\) have the same determinant up to scalar multiplication.

In this section, we use the \emph{Popov form} \cite{popov_properties_1970}, which can be defined for any matrix and with ``shifts''; here we only need the unshifted, nonsingular square case.
\begin{definition}
  For any row vector $\vec{v} \in \field[x]^{1 \times \delta}$ its \emph{row degree} denoted $\deg \vec{v}$ is the maximal degree among its entries.
  The \emph{pivot} of $\vec{v}$ is the rightmost entry of $\vec{v}$ with degree $\deg \vec{v}$.
  A nonsingular matrix $P = [ p_{ij} ] \in \field[x]^{\delta \times \delta}$ is in \emph{Popov form} if $p_{ii}$ is the pivot of the $i$th row, is monic, and $\deg p_{ii} > \deg p_{ji}$ for any $j \neq i$.
\end{definition}

For a (free) $\field[x]$-submodule $\M \subset \field[x]^{1 \times \delta}$ of rank \(\delta\), we identify a basis of $\M$ as the rows of a nonsingular matrix in $\field[x]^{\delta \times \delta}$.
Any such $\M$ has a unique basis $P \in \field[x]^{\delta \times \delta}$ in Popov form, which we call \emph{the Popov basis of $\M$}.
It has minimal row degrees in the following sense:
if $N \in \field[x]^{\delta \times \delta}$ is another basis of $\M$, there is a bijection $\psi$ from the rows of $P$ to the rows of $N$ such that $\deg \vec{p} \leq \deg \psi(\vec{p})$ for any row $\vec{p}$ of $P$.
The Popov basis satisfies $\degdet(\M) = \degdet(P) = |\!\cdeg(P)|$, using the following notation:
the sum of the entries of a tuple $\vec{t} \in \ZZ_{\geq 0}^\delta$ is denoted $|\vec{t}|$;
the column degree of a matrix $B \in \field[x]^{\delta \times \delta}$ is $\cdeg(B) = (d_i)_{i=1}^{\delta} \in \ZZ_{\geq 0}^{\delta}$, with $d_i$ the largest degree in the $i$th column of $B$ (for a zero column, $d_i =0$).

The next result allows us to compute Popov forms efficiently.

\begin{proposition}[\cite{neiger_computing_2017}]
  \label{prop:cost-popov}
  There is an algorithm which inputs a nonsingular matrix $B \in \field[x]^{\delta \times \delta}$ and outputs the Popov basis of the \(\field[x]\)-row space of $B$ using $\softO{\delta^{\expmm-1} |\!\cdeg(B)|}$ operations in $\field$, assuming that $\delta \in \bigO{|\!\cdeg(B)|}$.
\end{proposition}

Since Popov forms are ``column reduced'', they are well suited for matrix division with remainder \cite[Thm.\,6.3-15]{kailath_linear_1980}:
if $P \in \field[x]^{\delta \times \delta}$ is the Popov basis of $\M$, then for any $\vec{v} \in \field[x]^{1 \times \delta}$ there is a unique $\vec{u} \in \vec{v} + \M$ such that $\cdeg(\vec{u}) < \cdeg(P)$ entrywise; we denote $\vec{u} = \vec{v} \rem P$.
Furthermore, $\vec{u}$ has minimal row degree among all vectors in $\vec v + \M$.
Such remainders can be computed efficiently:

\begin{proposition}[\cite{neiger_computing_2017}]
  \label{prop:cost-rem-P}
  There is an algorithm which inputs a Popov form $P \in \field[x]^{\delta \times \delta}$ and $\vec{v} \in \field[x]^{1 \times \delta}$ such that $\cdeg(\vec{v}) < \cdeg(P) + (\degdet(P),\ldots,\degdet(P))$ entrywise, and outputs $\vec{v} \rem P$ using $\softO{\delta^{\expmm-1}\degdet(P)}$ operations in $\field$, assuming that $\delta \in \bigO{\degdet(P)}$.
\end{proposition}

\begin{algorithm}[ht]
  \caption{$\Precompute(G, \eta, \delta)$}
  \label{algo:Precompute}
  \begin{algorithmic}[1]
    \Input A reduced \(\ordLex\)-\Grobner basis $G = \set{b_0,\dots,b_s} \subset \field[x,y]$, sorted by increasing $y$-degree, for a zero-dimensional ideal $I$ (hence $b_0 \in \field[x]$);
    $\degTop, \delta \in \ZZ_{>0}$ with $\delta < \degTop$.
    \Output If no polynomial in \(y^{\degTop}+I\) has $y$-degree \(< \delta\), $\Fail$;
      otherwise, $g = y^{\degTop} - \hat{g} \in I$ with $\ydeg \hat{g} < \delta$ and $\xdeg \hat{g}$ minimal.
    \State $R \assign y^{\degTop} \rem G$ \label{line:precompute_rem} \Ifline{$\ydeg R \geq \delta$}{\Return $\Fail$} \label{line:precompute_fail}
    \State $B_\delta \assign $ basis of $I_\delta = \{ f \in I \mid \ydeg f < \delta \}$ as in \cref{cor:module_ideal}
    \State $B \in \field[x]^{\delta \times \delta} \assign $ row-wise applying $\FxIso_\delta$ to elements of $B_\delta$
    \State $P \in \field[x]^{\delta \times \delta} \assign$ Popov basis of $I_\delta$ from the basis $B$ \label{line:precompute_popov}
    \State $\hat{g} \assign -\FxIso_{\delta}^{-1} (\FxIso_{\delta}(R) \rem P) \in \field[x,y]$ \label{line:precompute_ghat}
    \State \Return $g = y^{\degTop} - \hat{g} \in \field[x,y]$
   \end{algorithmic}
\end{algorithm}

\begin{theorem}
  \label{thm:Precompute}
  \cref{algo:Precompute} is correct.
  Assuming $\degTop \in \bigO{\degdet(I_\delta)}$, it costs $\softO{\delta^{\omega-1} \degdet(I_\delta) + \degTop s \xdeg b_0}$ operations in \(\field\).
\end{theorem}

\begin{proof}
  Since $G$ is a \(\ordLex\)-\Grobner basis, if \(y^\degTop+I\) contains a polynomial of $y$-degree less than \(\delta\), then $\ydeg(y^\eta \rem G) \leq \delta$ and the algorithm does not fail at \cref{line:precompute_fail}.

  For correctness of the output, observe that $y^\eta - R \in I$ so satisfactory $g = y^\eta - \tilde g$ all have $\tilde g \in R + I_\delta$.
  Now, $\hat g$ of \cref{line:precompute_ghat} is clearly in $R + I_\delta$ since $P$ is the Popov basis of $I_\delta$, but also $\hat g$ has minimal $x$-degree in the coset $R + I_\delta$.
  Hence among all $g$ of the correct form, the algorithm returns that of minimal $x$-degree.

  For complexity, work is done in \cref{line:precompute_rem,line:precompute_popov,line:precompute_ghat}.
  Since $G$ is reduced, $\xdeg b_0 > \ldots > \xdeg b_s$.
  Therefore the diagonal entries in $B$ are dominant in their columns and $|\!\cdeg B| = \degdet(B) = \degdet(P) = \degdet(I_\delta)$.
  For \cref{line:precompute_rem}, we use the algorithm of \cite{van_der_hoeven_complexity_2015} with cost $\softO{\eta s \xdeg b_0}$, see \cref{lem:fast_rem}.
  \cref{line:precompute_popov} costs $\softO{\delta^{\omega-1}|\!\cdeg B|}$ by \cref{prop:cost-popov} and \cref{line:precompute_ghat} costs $\softO{\delta^{\omega-1}\degdet(P)}$ since $\xdeg R < \xdeg b_0 < \degdet(P)$.
\end{proof}

\subsection{Reshapers for the considered problems}
\label{ssec:precomp_appl}

We turn to obtaining the reduced $\ordLex$-\Grobner basis of $\VanIdeal(\pts)$.
We will consider the $\field[x]$-submodule $\VanIdeal_m(\pts) = \VanIdeal(\pts) \cap \field[x,y]_{\ydeg < m}$ which by \cref{lem:gb_vanideal,cor:module_ideal} is free and of rank $m$.
To obtain a \(\ordLex\)-\Grobner basis, our approach is to first compute the Hermite basis of $\VanIdeal_m(\pts)$.
This is the unique basis whose corresponding matrix $H \subset \field[x]^{m \times m}$ is lower triangular, with each diagonal entry monic and strictly dominating the degrees in its column.

\begin{lemma}
  \label{lem:gb_vanideal2}
  For any point set $\pts \subseteq \field^2$ and any \(m > \valX(\pts)\), we have \(\VanIdeal(\pts) = \ideal{\VanIdeal_m(\pts)}\) and \(\degdet(\VanIdeal_m(\pts)) = |\pts|\).
\end{lemma}
\begin{proof}
  By \cref{lem:gb_vanideal} the elements of the reduced \(\ordLex\)-Gr\"obner basis of \(\VanIdeal(\pts)\) have \(y\)-degree at most \(\valX(\pts)\), implying the first claim.
  Further, this means the quotient \(\field[x,y] / \VanIdeal(\pts)\) is isomorphic to the quotient of modules \(\field[x,y]_{\ydeg < m} / \VanIdeal_m(\pts)\).
  It is a basic property of zero-dimensional varieties that the $\field$-dimension of the former is the number of points in \(\pts\), which is hence also the $\field$-dimension of the latter.
  This dimension is \(\degdet(\VanIdeal_m(\pts))\) by \cite[Lem.\,2.3]{neiger_computing_2017}.
\end{proof}

\begin{proposition}
  \label{prop:compute_gb_vanideal}
  There is an algorithm which inputs \(\pts \subset \field^2\) and outputs the reduced $\ordLex$-\Grobner basis of $\VanIdeal(\pts)$ and has complexity \(\softO{\valX(\pts)^{\expmm-1} |\pts|}\).
\end{proposition}
\begin{proof}
  Let $\VanIdeal = \VanIdeal(\pts)$, $\VanIdeal_m = \VanIdeal_m(\pts)$, and $m = \valX(\pts) + 1$.
  We first compute the Hermite basis $H$ of $\VanIdeal_m(\pts)$ in time $\softO{m^{\expmm-1} |\pts|}$ using (a special case of) \cite[Thm.\,1.5]{jeannerod_fast_2016}, in which taking $\vec s = (0,n,\ldots,(m-1)n)$ ensures that the $\vec s$-Popov basis $P$ of $\VanIdeal_m$ is the Hermite basis.

  Let $G = \{ g_0,\ldots,g_{m-1} \} \subset \field[x,y]$ be given as the $\FxIso_m^{-1}$-image of the rows of $H$.
  By \cref{lem:gb_vanideal2} and since $H$ is lower triangular, $G$ is a $\ordLex$-\Grobner basis of $\VanIdeal$ but not necessarily minimal.
  Construct $G' \subseteq G$ from \(G\) by excluding the elements $g \in G$ such that there is $g' \in G$ with $\ydeg g' < \ydeg g$ and $\xdeg(\LCy(g')) \leq \xdeg(\LCy(g))$, i.e.~\(\LTlex(g')\) divides \(\LTlex(g)\).
  This makes \(G'\) a minimal \(\ordLex\)-Gr\"obner basis of \(\VanIdeal\) \cite[Lem.\,3 of Chap.\,2 \S 7]{CoxLittleOShea2015}, and we claim it is the reduced one.
  Indeed, since $H$ is in Hermite form, the selection criteria for $G'$ ensures that for any $g \neq g'$ in $G'$ and any term $x^iy^j$ in $g'$, we have $i < \xdeg(\LTlex(g))$ or $j < \ydeg g$, and hence $G'$ is reduced.
  Obtaining $G'$ from $H$ costs no arithmetic operations.
\end{proof}

\begin{corollary}
  \label{cor:precompute_pts}
  Given a point set $\pts \subseteq \field^2$ of cardinality $n$ and a reshaping sequence $\degSeq=(\degTop_i)_{i=0}^k$ with $n \geq \degTop_k$ and satisfying the condition of \cref{lem:exist_reshapers}, then we can determine if $\pts$ is $\degSeq$-balanced and compute an $\degSeq$-reshaper $\gPols = (g_i)_{i=1}^k$ for $\pts$ where each element has minimal possible $x$-degree in complexity $\softO{k\degTop_0^{\omega-1} n + \degTop_0 \valX n k}$.
\end{corollary}
\begin{proof}
  By \cref{prop:compute_gb_vanideal}, computing a reduced \(\ordLex\)-\Grobner basis $G = (b_i)_{i=0}^\valX$ of $\VanIdeal(\pts)$ costs $\softO{\nu_x^{\expmm-1} n} \subset \softO{\degTop_0^{\expmm-1} n}$.
  We then run \cref{algo:Precompute} on input $\degTop = \degTop_i$ and $\delta_i = 2\degTop_i - \degTop_{i-1} + 1 > \valX$ for $i=1,\ldots,k$.
  \cref{lem:gb_vanideal2} ensures $\degdet(\VanIdeal_{\delta}(\pts)) = n$ for any $\delta > \valX$,
  thus the cost of each call to \cref{algo:Precompute} becomes $\softO{\degTop_0^{\omega-1} n + \degTop_0 \valX n}$.
\end{proof}

\begin{corollary}
  Given $M, A \in \field[x]$ with $n := \deg M > \deg A$ and a reshaping sequence $\degSeq=(\degTop_i)_{i=0}^k$ with $n \geq \degTop_k$, then we can determine if $I := \ideal{M, y - A}$ is $\degSeq$-balanced and compute an $\degSeq$-reshaper $\gPols = (g_i)_{i=1}^k$ for $\pts$ where each element has minimal possible $x$-degree in complexity $\softO{k \degTop_0^{\omega-1} n}$.
\end{corollary}
\begin{proof}
  For any $\delta$, and using the notation of \cref{algo:Precompute}, the basis $B$ of $I_\delta$ is lower triangular with diagonal entries $(M,1,\ldots,1)$.
  Hence $\degdet(B) = \degdet(I_\delta) = n$.
  Using $s = 1$ and $\xdeg b_0 = \xdeg M = n$, the cost follows from \cref{thm:Precompute}.
\end{proof}

\section{Genericity}
\label{sec:genericity}

Now we show that on random input our algorithms usually have quasi-linear complexity, i.e.~that random point sets are balanced and that $\ideal{M, y-A}$ is balanced for random univariate $A, M$.

\begin{lemma}
  \label{lem:trancsend_eval_matrix}
  Let $\alpha_1,\ldots,\alpha_n \in \field$ be distinct, let $y_1,\ldots,y_n$ be new indeterminates,
  and consider for $s \in \ZZpos$ the matrix
  \begin{equation}
    \label{eqn:generic_kernel}
    A_s = \big[ V_s \mid D V_s \mid \ldots \mid D^{m-1} V_s \big] \in \field[y_1,\ldots,y_n]^{n \times ms}
  \end{equation}
  where $D$ is the diagonal matrix with entries $(y_1,\ldots,y_n)$, and $V_s = [ \alpha_i^{j-1} ]_{1\le i\le n,1 \le j\le s} \in \field^{n \times s}$.
  Then $A_s$ has rank $\min(n, ms)$.
\end{lemma}
\begin{proof}
  Note that by rank of a matrix over \(\field[y_1,\ldots,y_n]\), we mean the rank of that matrix seen as over the field of fractions \(\field(y_1,\ldots,y_n)\).
  If we specialise $y_i$ to $\alpha_i^s$ for $i = 1, \ldots, n$, we obtain the Vandermonde matrix $\hat A_s = [ \alpha_i^{j-1} ]_{1 \le i\le n,1 \le j \le ms} \in \field^{n\times ms}$ of the points $\alpha_1,\ldots,\alpha_n$.
  Since these points are distinct, \(\hat A_s\) has full rank $\min(n, ms)$.
  Hence $A_s$ must also have full rank.
\end{proof}

The columns of $A_s$ can be identified to monomials $x^i y^j$ with \(i<s\) and \(j<m\).
In particular, if $p \in \VanIdeal(\pts)$ is a bivariate polynomial with $x$-degree less than $s$ and $y$-degree less than $m$ which vanishes on a point set $\pts = \{ (\alpha_i,\beta_i) \}_{i=1}^n \subset \field^2$ with distinct \(\alpha_i\)'s, then the coefficients of $p$ form a vector in the right kernel of the matrix $\hat A_s = (A_s)_{|y_i \rightarrow \beta_i} \in \field^{n \times ms}$ specializing $y_i$ to $\beta_i$.

The next lemma determines the exact row degrees of the Popov basis $P \in \field[x]^{m \times m}$ of $\FxIso_m(\VanIdeal_m(\pts))$ for a ``random'' point set \(\pts\), where $\VanIdeal_m(\pts) = \VanIdeal(\pts) \cap \field[x,y]_{\ydeg < m}$ as in \cref{ssec:precomp_appl}.

\begin{lemma}
  \label{lem:gen_small_popov}
  Let $\alpha_1,\ldots,\alpha_n \in \field$ be distinct, let $\T \subseteq \field$ be a finite subset, and let $\lambda: \field^n \rightarrow \field^n$ be an affine map.
  For $\gamma_1,\ldots,\gamma_n \in \T$ chosen independently and uniformly at random, set $\pts = \{ (\alpha_i,\beta_i) \}_{i=1}^n$ where $(\beta_1,\ldots,\beta_n) = \lambda(\gamma_1,\ldots,\gamma_n)$.
  Let $m \in \ZZ$ with $\valX(\pts) < m \leq n$ and let $(d,t) = \quorem(n, m)$.
  With probability at least $1 - 2nm/|\T|$, the Popov basis $P \in \field[x]^{m \times m}$ of $\FxIso_{m}(\VanIdeal_m(\pts))$ has exactly $m-t$ rows of degree $d$ and $t$ rows of degree $d+1$ and
  in particular $\deg_x P \leq d+1$.
\end{lemma}
\begin{proof}
  Let $p_1,\ldots,p_m \in \field[x,y]$ be the polynomials defined by the rows of $P$.
  \cref{lem:gb_vanideal} shows $\degdet(P) = n = \sum_{i=1}^m \deg_x p_i$.

  For any $s \in \ZZpos$, let $A_s \in \field[y_1,\ldots,y_n]^{n \times ms}$ be as in \cref{lem:trancsend_eval_matrix}, hence $\rank(A_s) = \min(n, ms)$.
  Let $\hat A_s = (A_s)_{|y_i \rightarrow \beta_i} \in \field^{n \times ms}$.
  Taking $s = d$, as mentioned above, if $\deg_x p_i < d$ for some $i$, then the coefficient vector of $p_i$ is in the right kernel of $\hat A_d$, and so $\rank(\hat A_d) < \rank(A_d) = md \le n$.
  Thus, letting $M \in \field[y_1,\ldots,y_n]$ be a non-zero $md \times md$ minor of $A_d$ then $M(\beta_1,\ldots,\beta_n) = M(\lambda(\gamma_1,\ldots,\gamma_n)) = 0$;
  $M$ has degree at most $m-1$ in each variable, so the total degree of $M$ is less than $nm$, and since \(\lambda\) is affine the composition $M(\lambda(z_1,\ldots,z_n))$ also has total degree less than $nm$.
  Then, by \cref{lem:zippel} the probability that $M(\lambda(\gamma_1,\ldots,\gamma_n)) = 0$ is at most $nm/|\T|$.

  Assume now that all rows of $P$ have degree at least $d$.
  For each \(i\) such that $\xdeg p_i = d$, the coefficients of $p_i$ form a vector in the right kernel of $\hat A_{d+1} \in \field^{n \times m(d+1)}$.
  By \cref{lem:trancsend_eval_matrix}, $A_{d+1}$ has a right kernel (over the fractions) of dimension $m(d+1) - n = m-t$.
  Since the rows of $P$ are linearly independent over $\field[x]$, and therefore also over $\field$, whenever $\rank(\hat A_{d+1}) = \rank(A_{d+1})$ at most $m-t$ rows of $P$ have $x$-degree $d$.
  We thus consider $N \in \field[y_1,\ldots,y_n]$ a non-zero $n \times n$ minor of $A_{d+1}$.
  Again $N$ has total degree less than $nm$ and the probability that $N(\beta_1,\ldots,\beta_n) = N(\lambda(\gamma_1,\ldots,\gamma_n)) = 0$ is at most $nm/|\T|$, bounding the probability that $\rank(\hat A_{d+1}) < \rank(A_{d+1})$.

  Hence, with probability at least $1 -2nd/|\T|$, $P$ has all rows of degree at least $d$ and $j$ rows of degree exactly $d$ with \(j \le m-t\).
  Each of the remaining $m-j$ rows has degree at least $d+1$, while their degrees must sum to
  \(
    n - jd = md + t - jd = (m-j)d + t \le (m-j)(d+1)
  \).
  Hence each of them has degree exactly $d+1$.
\end{proof}

\cref{algo:Precompute} for computing reshapers outputs a $g = y^\eta - \hat g$ with $\ydeg \hat g < \delta$ satisfying $\xdeg \hat g \leq \xdeg P$, where $P$ is the Popov basis of $\VanIdeal_\delta(\pts)$.
\cref{lem:gen_small_popov} states that generically we can expect $\xdeg P \le \floor{\frac n \delta} + 1$, and so when $\delta = 2\eta_i - \eta_{i-1} + 1$ in a reshaping sequence, this matches the definition of $\degSeq$-balanced.

\begin{corollary}
  \label{prop:gen_pts_balanced}
  Let $\alpha_1,\ldots,\alpha_n \in \field$ be distinct, let $\T \subseteq \field$ a finite subset, and let $\lambda: \field^n \rightarrow \field^n$ be an affine map.
  For $\gamma_1,\ldots,\gamma_n \in \T$ chosen independently and uniformly at random, set $\pts = \{ (\alpha_i,\beta_i) \}_{i=1}^n$ where $(\beta_1,\ldots,\beta_n) = \lambda(\gamma_1,\ldots,\gamma_n)$.
  Let $\degSeq = ( \degTop_i )_{i=0}^k$ be a reshaping sequence with $\degTop_0 \leq n$ and satisfying the constraint of \cref{lem:exist_reshapers}.
  Then $\pts$ is $\degSeq$-balanced with probability at least $1 - n^2k/|\T|$.
\end{corollary}

The above proposition directly applies to both our MPE and interpolation algorithms on random point sets with unique $x$-coordinates.
Note that in the case of interpolation, where the point set is sheared if its $y$-valency is greater than one, the property of being $\degSeq$-balanced is not inherited a priori by the sheared point set. The probability of being $\degSeq$-balanced, however, is preserved, since the shearing acts as an affine transformation on the $y$-coordinates.
There are many formulations depending on the type of randomness one needs over the point sets; the following is a simple example over finite fields:

\begin{corollary}
  \label{cor:prob_complexity_mpe_interp}
  Let $d, n \in \ZZpos$ with $d \leq n$ and $\FF q$ be a finite field with $q$ elements, and let $\pts = \{ (\alpha_i,\beta_i) \}_{i=1}^n \subseteq \FF q^2$ be chosen uniformly at random among point sets with cardinality $n$.
  Then with probability of at least $\big(1 - \frac{n^2} q\big)\big(1 - \frac{3n^2(\log_{3/2}(n)+1)}{q}\big)$ over the choice of $\pts$ the following two problems can be solved with cost $\softO{n}$:
  \begin{enumerate}
    \item Input polynomial $f \in \FF q[x,y]$ such that $\xdeg f < n/d$ and $\ydeg f < d$, and output $( f(\alpha_i, \beta_i) )_{i=1}^n \in \FF q^n$.
    \item Input interpolation values $\vec\gamma = ( \gamma_i )_{i=1}^n \in \FF q^n$, and output $f \in \FF q[x,y]$ satisfying $f(\alpha_i, \beta_i) = \gamma_i$ for $i=1,\ldots,n$, as well as $\deg_y f < d$ and $\deg_x f \leq cn$ for some constant $c$ which depends only on $n$ and $d$.
  \end{enumerate}
\end{corollary}
\begin{proof}[Proof sketch]
  The probability simply bounds the probability that $\pts$ has unique $x$-coordinates \emph{and} that it is balanced in all the necessary ways.
  By \cref{cor:reshapers_constrained} there is an appropriate reshaping sequence of length at most $\log_{3/2}(n) + 2$.
\end{proof}

We do not make a claim about the genericity in \cref{algo:ValencyMPE}: due to the shearing in that algorithm, the arguments of this section do not immediately apply.
Lastly, we turn to modular composition.
\begin{theorem}
  Let $M \in \field[x]$ be square-free of degree $n$ and let $\degSeq$ be a $(d, 1)$-reshaping sequence of length $k$ with $0 < d \leq n$.
  Let $\T \subseteq \field$ be a finite subset, and let $A = \sum_{i=0}^{n-1} a_i x^{i-1} \in \field[x]$ where $a_0,\ldots,a_{n-1}$ are chosen independently and uniformly at random from \(\T\).
  Then $\ideal{M, y - A}$ is $\degSeq$-balanced with probability at least $1 - n^2k/|\T|$.
\end{theorem}
\begin{proof}
  Let $\ext$ be the splitting field of $M$, so $M = \prod_{i=1}^n (x-\alpha_i)$ for some pairwise distinct $\alpha_1,\ldots,\alpha_n \in \ext$.
  Define the stochastic variables $\beta_i = A(\alpha_i)$ for $i=1,\ldots,n$;
  the map $\lambda(a_0,\ldots,a_{n-1}) = (\beta_1,\ldots,\beta_n)$ is $\ext$-linear.
  Consider $\pts = \{ (\alpha_i, \beta_i) \}_{i=1}^n \subseteq \ext^2$.
  Then \cref{prop:gen_pts_balanced} implies that $\pts$ is $\degSeq$-balanced with probability at least $1 - \frac{n^2k}{|\T|}$.
  In this case, for each $i$ there exists $g_i = y^{\degTop_i} + \hat{g}_i \in I_{\ext}$ where $\ydeg \hat{g}_i < 2\degTop_i - \degTop_{i-1}$ and $\xdeg \hat{g}_i \leq \floor{\frac n {2\degTop_i - \degTop_{i-1}+1}} + 1$, and where $I_\ext = \ideal{M, y-A} \otimes_\field \ext$.
  Let $\{1,\zeta,\ldots,\zeta^{s-1} \} \subset \ext$ be a basis of $\ext : \field$ and write $g_i = g_{i,0} + \zeta  g_{i,1} + \ldots + \zeta^{s-1}  g_{i,s-1}$ with $ g_{i,j} \in \field[x,y]$.
  Then $g_i \in I_{\ext}$ implies that $ g_{i,0} \in I$, and by the shape of $g_i$ then $ g_{i,0} = y^{\eta_i} +  \hat g_{i,0}$ where the $x$- and $y$-degree of $\hat g_{i,0}$ satisfy the same bounds as $\hat{g_i}$.
  Then the tuple $\vec g_0 = (  g_{1,0}, \ldots, g_{k,0} ) \in \field[x,y]^k$ forms a balanced $\degSeq$-reshaper for $I$.
\end{proof}

\bibliographystyle{ACM-Reference-Format}

\begin{thebibliography}{31}


\ifx \showCODEN    \undefined \def \showCODEN     #1{\unskip}     \fi
\ifx \showDOI      \undefined \def \showDOI       #1{#1}\fi
\ifx \showISBNx    \undefined \def \showISBNx     #1{\unskip}     \fi
\ifx \showISBNxiii \undefined \def \showISBNxiii  #1{\unskip}     \fi
\ifx \showISSN     \undefined \def \showISSN      #1{\unskip}     \fi
\ifx \showLCCN     \undefined \def \showLCCN      #1{\unskip}     \fi
\ifx \shownote     \undefined \def \shownote      #1{#1}          \fi
\ifx \showarticletitle \undefined \def \showarticletitle #1{#1}   \fi
\ifx \showURL      \undefined \def \showURL       {\relax}        \fi
\providecommand\bibfield[2]{#2}
\providecommand\bibinfo[2]{#2}
\providecommand\natexlab[1]{#1}
\providecommand\showeprint[2][]{arXiv:#2}

\bibitem[\protect\citeauthoryear{Assmus and Key}{Assmus and Key}{1992}]%
        {assmus_designs_1992}
\bibfield{author}{\bibinfo{person}{E.~F. Assmus} {and} \bibinfo{person}{J.~D.
  Key}.} \bibinfo{year}{1992}\natexlab{}.
\newblock \bibinfo{booktitle}{\emph{Designs and {Their} {Codes}}}.
\newblock \bibinfo{publisher}{Cambridge University Press}.
\newblock
\showISBNx{978-0-521-41361-9}
\urldef\tempurl%
\url{https://doi.org/10.1017/CBO9781316529836}
\showDOI{\tempurl}


\bibitem[\protect\citeauthoryear{Brent and Kung}{Brent and Kung}{1978}]%
        {brent-kung-1978}
\bibfield{author}{\bibinfo{person}{R.~P. Brent} {and} \bibinfo{person}{H.~T.
  Kung}.} \bibinfo{year}{1978}\natexlab{}.
\newblock \showarticletitle{Fast Algorithms for Manipulating Formal Power
  Series}.
\newblock \bibinfo{journal}{\emph{J. ACM}} \bibinfo{volume}{25},
  \bibinfo{number}{4} (\bibinfo{year}{1978}), \bibinfo{pages}{581--595}.
\newblock
\showISSN{0004-5411}
\urldef\tempurl%
\url{https://doi.org/10.1145/322092.322099}
\showDOI{\tempurl}


\bibitem[\protect\citeauthoryear{Cantor and Kaltofen}{Cantor and
  Kaltofen}{1991}]%
        {cantor_fast_1991}
\bibfield{author}{\bibinfo{person}{D.~G. Cantor} {and} \bibinfo{person}{E.
  Kaltofen}.} \bibinfo{year}{1991}\natexlab{}.
\newblock \showarticletitle{On fast multiplication of polynomials over
  arbitrary algebras}.
\newblock \bibinfo{journal}{\emph{Acta Informatica}} \bibinfo{volume}{28},
  \bibinfo{number}{7} (\bibinfo{year}{1991}), \bibinfo{pages}{693--701}.
\newblock
\showISSN{0001-5903, 1432-0525}
\urldef\tempurl%
\url{https://doi.org/10.1007/BF01178683}
\showDOI{\tempurl}


\bibitem[\protect\citeauthoryear{Cox, Little, and O'Shea}{Cox
  et~al\mbox{.}}{2015}]%
        {CoxLittleOShea2015}
\bibfield{author}{\bibinfo{person}{D.~A. Cox}, \bibinfo{person}{J. Little},
  {and} \bibinfo{person}{D. O'Shea}.} \bibinfo{year}{2015}\natexlab{}.
\newblock \bibinfo{booktitle}{\emph{Ideals, {Varieties}, and {Algorithms}}
  (\bibinfo{edition}{4th} ed.)}.
\newblock \bibinfo{publisher}{Springer}.
\newblock
\urldef\tempurl%
\url{https://doi.org/10.1007/978-3-319-16721-3}
\showDOI{\tempurl}


\bibitem[\protect\citeauthoryear{Coxon}{Coxon}{2018}]%
        {coxon_fast_2018}
\bibfield{author}{\bibinfo{person}{N. Coxon}.} \bibinfo{year}{2018}\natexlab{}.
\newblock \showarticletitle{Fast systematic encoding of multiplicity codes}.
\newblock \bibinfo{journal}{\emph{J. Symb. Comput.}} (\bibinfo{year}{2018}).
\newblock
\showISSN{0747-7171}
\urldef\tempurl%
\url{https://doi.org/10.1016/j.jsc.2018.08.005}
\showDOI{\tempurl}


\bibitem[\protect\citeauthoryear{Dahan}{Dahan}{2009}]%
        {Dahan2009}
\bibfield{author}{\bibinfo{person}{X. Dahan}.} \bibinfo{year}{2009}\natexlab{}.
\newblock \showarticletitle{Size of Coefficients of Lexicographical Gr{\"o}bner
  Bases: The Zero-Dimensional, Radical and Bivariate Case}. In
  \bibinfo{booktitle}{\emph{Proceedings ISSAC 2009}}.
  \bibinfo{pages}{119--126}.
\newblock
\urldef\tempurl%
\url{https://doi.org/10.1145/1576702.1576721}
\showDOI{\tempurl}


\bibitem[\protect\citeauthoryear{DeMillo and Lipton}{DeMillo and
  Lipton}{1978}]%
        {DeMilloLipton78}
\bibfield{author}{\bibinfo{person}{R.~A. DeMillo} {and} \bibinfo{person}{R.~J.
  Lipton}.} \bibinfo{year}{1978}\natexlab{}.
\newblock \showarticletitle{A Probabilistic Remark on Algebraic Program
  Testing}.
\newblock \bibinfo{journal}{\emph{Inf. Process. Lett.}} \bibinfo{volume}{7},
  \bibinfo{number}{4} (\bibinfo{year}{1978}), \bibinfo{pages}{193--195}.
\newblock
\urldef\tempurl%
\url{https://doi.org/10.1016/0020-0190(78)90067-4}
\showDOI{\tempurl}


\bibitem[\protect\citeauthoryear{Hart, Johansson, and Pancratz}{Hart
  et~al\mbox{.}}{2015}]%
        {flint}
\bibfield{author}{\bibinfo{person}{W. Hart}, \bibinfo{person}{F. Johansson},
  {and} \bibinfo{person}{S. Pancratz}.} \bibinfo{year}{2015}\natexlab{}.
\newblock \bibinfo{title}{{FLINT}: {F}ast {L}ibrary for {N}umber {T}heory}.
\newblock
\newblock
\newblock
\shownote{Version 2.5.2, \url{http://flintlib.org}.}


\bibitem[\protect\citeauthoryear{Jeannerod, Neiger, Schost, and
  Villard}{Jeannerod et~al\mbox{.}}{2016}]%
        {jeannerod_fast_2016}
\bibfield{author}{\bibinfo{person}{C.-P. Jeannerod}, \bibinfo{person}{V.
  Neiger}, \bibinfo{person}{\'E. Schost}, {and} \bibinfo{person}{G. Villard}.}
  \bibinfo{year}{2016}\natexlab{}.
\newblock \showarticletitle{Fast {Computation} of {Minimal} {Interpolation}
  {Bases} in {Popov} {Form} for {Arbitrary} {Shifts}}. In
  \bibinfo{booktitle}{\emph{Proceedings ISSAC 2016}}.
  \bibinfo{pages}{295--302}.
\newblock
\showISBNx{978-1-4503-4380-0}
\urldef\tempurl%
\url{https://doi.org/10.1145/2930889.2930928}
\showDOI{\tempurl}


\bibitem[\protect\citeauthoryear{Kailath}{Kailath}{1980}]%
        {kailath_linear_1980}
\bibfield{author}{\bibinfo{person}{T Kailath}.}
  \bibinfo{year}{1980}\natexlab{}.
\newblock \bibinfo{booktitle}{\emph{Linear {Systems}}}.
\newblock \bibinfo{publisher}{Prentice-Hall}.
\newblock


\bibitem[\protect\citeauthoryear{Kedlaya and Umans}{Kedlaya and Umans}{2011}]%
        {kedlaya_fast_2011}
\bibfield{author}{\bibinfo{person}{K. Kedlaya} {and} \bibinfo{person}{C.
  Umans}.} \bibinfo{year}{2011}\natexlab{}.
\newblock \showarticletitle{Fast {Polynomial} {Factorization} and {Modular}
  {Composition}}.
\newblock \bibinfo{journal}{\emph{SIAM J. Comput.}} \bibinfo{volume}{40},
  \bibinfo{number}{6} (\bibinfo{date}{Jan.} \bibinfo{year}{2011}),
  \bibinfo{pages}{1767--1802}.
\newblock
\showISSN{0097-5397}
\urldef\tempurl%
\url{https://doi.org/10.1137/08073408X}
\showDOI{\tempurl}


\bibitem[\protect\citeauthoryear{Lazard}{Lazard}{1985}]%
        {lazard_ideal_1985}
\bibfield{author}{\bibinfo{person}{D. Lazard}.}
  \bibinfo{year}{1985}\natexlab{}.
\newblock \showarticletitle{Ideal bases and primary decomposition: case of two
  variables}.
\newblock \bibinfo{journal}{\emph{J. Symb. Comput.}} \bibinfo{volume}{1},
  \bibinfo{number}{3} (\bibinfo{year}{1985}).
\newblock
\urldef\tempurl%
\url{https://doi.org/10.1016/S0747-7171(85)80035-3}
\showDOI{\tempurl}


\bibitem[\protect\citeauthoryear{Le~Gall}{Le~Gall}{2014}]%
        {le_gall_powers_2014}
\bibfield{author}{\bibinfo{person}{F. Le~Gall}.}
  \bibinfo{year}{2014}\natexlab{}.
\newblock \showarticletitle{Powers of tensors and fast matrix multiplication}.
  In \bibinfo{booktitle}{\emph{Proceedings ISSAC 2014}}.
  \bibinfo{publisher}{ACM}, \bibinfo{pages}{296--303}.
\newblock
\urldef\tempurl%
\url{https://doi.org/10.1145/2608628.2608664}
\showDOI{\tempurl}


\bibitem[\protect\citeauthoryear{Miura}{Miura}{1993}]%
        {miura_algebraic_1993}
\bibfield{author}{\bibinfo{person}{S. Miura}.} \bibinfo{year}{1993}\natexlab{}.
\newblock \showarticletitle{Algebraic geometric codes on certain plane curves}.
\newblock \bibinfo{journal}{\emph{Electronics and Communications in Japan (Part
  III: Fundamental Electronic Science)}} \bibinfo{volume}{76},
  \bibinfo{number}{12} (\bibinfo{year}{1993}), \bibinfo{pages}{1--13}.
\newblock
\showISSN{1520-6440}
\urldef\tempurl%
\url{https://doi.org/10.1002/ecjc.4430761201}
\showDOI{\tempurl}


\bibitem[\protect\citeauthoryear{Neiger, Salvy, Schost, and Villard}{Neiger
  et~al\mbox{.}}{2020}]%
        {NeigerSalvySchostVillard}
\bibfield{author}{\bibinfo{person}{V. Neiger}, \bibinfo{person}{B. Salvy},
  \bibinfo{person}{{\'E}. Schost}, {and} \bibinfo{person}{G. Villard}.}
  \bibinfo{year}{2020}\natexlab{}.
\newblock \bibinfo{title}{Faster modular composition (work in progress)}.
\newblock
\newblock


\bibitem[\protect\citeauthoryear{Neiger and Vu}{Neiger and Vu}{2017}]%
        {neiger_computing_2017}
\bibfield{author}{\bibinfo{person}{V. Neiger} {and} \bibinfo{person}{T.~X.
  Vu}.} \bibinfo{year}{2017}\natexlab{}.
\newblock \showarticletitle{Computing {Canonical} {Bases} of {Modules} of
  {Univariate} {Relations}}. In \bibinfo{booktitle}{\emph{Proceedings ISSAC
  2017}}.
\newblock
\urldef\tempurl%
\url{https://doi.org/10.1145/3087604.3087656}
\showDOI{\tempurl}


\bibitem[\protect\citeauthoryear{N{\"u}sken and Ziegler}{N{\"u}sken and
  Ziegler}{2004}]%
        {nusken-ziegler-2004}
\bibfield{author}{\bibinfo{person}{M. N{\"u}sken} {and} \bibinfo{person}{M.
  Ziegler}.} \bibinfo{year}{2004}\natexlab{}.
\newblock \showarticletitle{Fast Multipoint Evaluation of Bivariate
  Polynomials}. In \bibinfo{booktitle}{\emph{Proceedings ESA 2004}}.
\newblock
\showISBNx{978-3-540-30140-0}
\urldef\tempurl%
\url{https://doi.org/10.1007/978-3-540-30140-0_49}
\showDOI{\tempurl}


\bibitem[\protect\citeauthoryear{Pan}{Pan}{1994}]%
        {pan_simple_1994}
\bibfield{author}{\bibinfo{person}{V.~Y. Pan}.}
  \bibinfo{year}{1994}\natexlab{}.
\newblock \showarticletitle{Simple {Multivariate} {Polynomial}
  {Multiplication}}.
\newblock \bibinfo{journal}{\emph{J. Symb. Comput.}} \bibinfo{volume}{18},
  \bibinfo{number}{3} (\bibinfo{year}{1994}), \bibinfo{pages}{183--186}.
\newblock
\showISSN{0747-7171}
\urldef\tempurl%
\url{https://doi.org/10.1006/jsco.1994.1042}
\showDOI{\tempurl}


\bibitem[\protect\citeauthoryear{Paterson and Stockmeyer}{Paterson and
  Stockmeyer}{1973}]%
        {paterson_number_1973}
\bibfield{author}{\bibinfo{person}{M.~S. Paterson} {and} \bibinfo{person}{L.~J.
  Stockmeyer}.} \bibinfo{year}{1973}\natexlab{}.
\newblock \showarticletitle{On the number of nonscalar multiplications
  necessary to evaluate polynomials}.
\newblock \bibinfo{journal}{\emph{SIAM J. Comput.}} \bibinfo{volume}{2},
  \bibinfo{number}{1} (\bibinfo{year}{1973}), \bibinfo{pages}{60--66}.
\newblock
\urldef\tempurl%
\url{https://doi.org/10.1137/0202007}
\showDOI{\tempurl}


\bibitem[\protect\citeauthoryear{Popov}{Popov}{1970}]%
        {popov_properties_1970}
\bibfield{author}{\bibinfo{person}{V Popov}.} \bibinfo{year}{1970}\natexlab{}.
\newblock \showarticletitle{Some properties of the control systems with
  irreducible matrix-transfer functions}. In \bibinfo{booktitle}{\emph{Seminar
  on {Diff.} {Eq.} and {Dyn.} {Sys.}, {II}}}. \bibinfo{pages}{169--180}.
\newblock
\urldef\tempurl%
\url{https://doi.org/10.1007/BFb0059934}
\showDOI{\tempurl}


\bibitem[\protect\citeauthoryear{Schwartz}{Schwartz}{1980}]%
        {Schwartz80}
\bibfield{author}{\bibinfo{person}{J.~T. Schwartz}.}
  \bibinfo{year}{1980}\natexlab{}.
\newblock \showarticletitle{Fast Probabilistic Algorithms for Verification of
  Polynomial Identities}.
\newblock \bibinfo{journal}{\emph{J. ACM}} \bibinfo{volume}{27},
  \bibinfo{number}{4} (\bibinfo{year}{1980}), \bibinfo{pages}{701--717}.
\newblock
\urldef\tempurl%
\url{https://doi.org/10.1145/322217.322225}
\showDOI{\tempurl}


\bibitem[\protect\citeauthoryear{Shoup}{Shoup}{2020}]%
        {ShoupNTL}
\bibfield{author}{\bibinfo{person}{V. Shoup}.} \bibinfo{year}{2020}\natexlab{}.
\newblock \bibinfo{title}{{NTL}: A Library for doing Number Theory, version
  11.4.3}.
\newblock \bibinfo{howpublished}{\texttt{http://www.shoup.net}}.
\newblock


\bibitem[\protect\citeauthoryear{van~der Hoeven}{van~der Hoeven}{2015}]%
        {van_der_hoeven_complexity_2015}
\bibfield{author}{\bibinfo{person}{J. van~der Hoeven}.}
  \bibinfo{year}{2015}\natexlab{}.
\newblock \showarticletitle{On the complexity of multivariate polynomial
  division}. In \bibinfo{booktitle}{\emph{Proceedings ACA 2015}}.
  \bibinfo{pages}{447--458}.
\newblock
\urldef\tempurl%
\url{https://doi.org/10.1007/978-3-319-56932-1_28}
\showDOI{\tempurl}


\bibitem[\protect\citeauthoryear{van~der Hoeven and Lecerf}{van~der Hoeven and
  Lecerf}{2018}]%
        {van_der_hoeven_modular_2018}
\bibfield{author}{\bibinfo{person}{J. van~der Hoeven} {and} \bibinfo{person}{G.
  Lecerf}.} \bibinfo{year}{2018}\natexlab{}.
\newblock \showarticletitle{Modular composition via factorization}.
\newblock \bibinfo{journal}{\emph{J. Complexity}}  \bibinfo{volume}{48}
  (\bibinfo{year}{2018}), \bibinfo{pages}{36--68}.
\newblock
\showISSN{0885-064X}
\urldef\tempurl%
\url{https://doi.org/10.1016/j.jco.2018.05.002}
\showDOI{\tempurl}


\bibitem[\protect\citeauthoryear{van~der Hoeven and Lecerf}{van~der Hoeven and
  Lecerf}{2019}]%
        {van_der_hoeven_fast_2019}
\bibfield{author}{\bibinfo{person}{J. van~der Hoeven} {and} \bibinfo{person}{G.
  Lecerf}.} \bibinfo{year}{2019}\natexlab{}.
\newblock \showarticletitle{Fast multivariate multi-point evaluation
  revisited}.
\newblock \bibinfo{journal}{\emph{J. Complexity}} (\bibinfo{year}{2019}).
\newblock
\showISSN{0885-064X}
\urldef\tempurl%
\url{https://doi.org/10.1016/j.jco.2019.04.001}
\showDOI{\tempurl}


\bibitem[\protect\citeauthoryear{van~der Hoeven and Schost}{van~der Hoeven and
  Schost}{2013}]%
        {van_der_hoeven_multi-point_2013}
\bibfield{author}{\bibinfo{person}{J. van~der Hoeven} {and}
  \bibinfo{person}{{\'E}. Schost}.} \bibinfo{year}{2013}\natexlab{}.
\newblock \showarticletitle{Multi-point evaluation in higher dimensions}.
\newblock \bibinfo{journal}{\emph{Appl. Algebra Eng. Commun. Comput.}}
  \bibinfo{volume}{24}, \bibinfo{number}{1} (\bibinfo{year}{2013}),
  \bibinfo{pages}{37--52}.
\newblock
\urldef\tempurl%
\url{https://doi.org/10.1007/s00200-012-0179-3}
\showDOI{\tempurl}


\bibitem[\protect\citeauthoryear{Villard}{Villard}{2018a}]%
        {villard_computing_2018}
\bibfield{author}{\bibinfo{person}{G. Villard}.}
  \bibinfo{year}{2018}\natexlab{a}.
\newblock \showarticletitle{On computing the resultant of generic bivariate
  polynomials}. In \bibinfo{booktitle}{\emph{Proceedings ISSAC 2018}}.
  \bibinfo{pages}{391--398}.
\newblock
\urldef\tempurl%
\url{https://doi.org/10.1145/3208976.3209020}
\showDOI{\tempurl}


\bibitem[\protect\citeauthoryear{Villard}{Villard}{2018b}]%
        {villard_computing_2018_slides}
\bibfield{author}{\bibinfo{person}{G. Villard}.}
  \bibinfo{year}{2018}\natexlab{b}.
\newblock \bibinfo{title}{On computing the resultant of generic bivariate
  polynomials. Presentation at ISSAC 2018}.
\newblock
\newblock
\urldef\tempurl%
\url{http://www.issac-conference.org/2018/slides/villard-computingresultant.pdf}
\showURL{%
\tempurl}


\bibitem[\protect\citeauthoryear{von~zur Gathen}{von~zur Gathen}{1990}]%
        {Gathen1990}
\bibfield{author}{\bibinfo{person}{J. von~zur Gathen}.}
  \bibinfo{year}{1990}\natexlab{}.
\newblock \showarticletitle{Functional decomposition of polynomials: the tame
  case}.
\newblock \bibinfo{journal}{\emph{J. Symb. Comput.}} \bibinfo{volume}{9},
  \bibinfo{number}{3} (\bibinfo{year}{1990}).
\newblock
\showISSN{0747-7171}
\urldef\tempurl%
\url{https://doi.org/10.1016/S0747-7171(08)80014-4}
\showDOI{\tempurl}


\bibitem[\protect\citeauthoryear{von~zur Gathen and Gerhard}{von~zur Gathen and
  Gerhard}{2013}]%
        {von_zur_gathen_modern_2012}
\bibfield{author}{\bibinfo{person}{J. von~zur Gathen} {and} \bibinfo{person}{J.
  Gerhard}.} \bibinfo{year}{2013}\natexlab{}.
\newblock \bibinfo{booktitle}{\emph{Modern {Computer} {Algebra}}
  (\bibinfo{edition}{3rd} ed.)}.
\newblock \bibinfo{publisher}{Cambridge University Press}.
\newblock
\urldef\tempurl%
\url{https://doi.org/10.1017/CBO9781139856065}
\showDOI{\tempurl}


\bibitem[\protect\citeauthoryear{Zippel}{Zippel}{1979}]%
        {Zippel79}
\bibfield{author}{\bibinfo{person}{R. Zippel}.}
  \bibinfo{year}{1979}\natexlab{}.
\newblock \showarticletitle{Probabilistic algorithms for sparse polynomials}.
  In \bibinfo{booktitle}{\emph{Proceedings EUROSAM'79}}.
  \bibinfo{pages}{216--226}.
\newblock
\urldef\tempurl%
\url{https://doi.org/10.1007/3-540-09519-5_73}
\showDOI{\tempurl}


\end{thebibliography}

\appendix
\renewcommand\thesection{}
\section{Appendix}
\renewcommand\thesection{A}

\begin{corollary}[of \protect{\cite{lazard_ideal_1985}}]
  \label{cor:lazard}
  Let $G = \{ b_0,\ldots,b_s\} \subset \field[x,y]$ be a minimal \(\ordLex\)-\Grobner basis, sorted according to \(\ordLex\).
  Then
  \begin{enumerate}
    \item $\ydeg b_0 < \ldots < \ydeg b_s$; and
    \item $\LCy(b_{s}) \mid \LCy(b_{s-1}) \mid \cdots \mid \LCy(b_0)$.
  \end{enumerate}
\end{corollary}

\begin{proof}[Proof of \cref{cor:module_ideal}]
  Since $I$ is an ideal of $\field[x,y]$ and \(I_\delta = I \cap \field[x,y]_{\ydeg<\delta}\), then $I_\delta$ is a $\field[x]$-submodule of $\field[x,y]_{\ydeg<\delta}$.
  Let \(\mathcal{B}\) denote the (claimed) basis in the corollary. Clearly $\mathcal{B} \subseteq I_\delta$, and the elements of $\mathcal{B}$ all have different $y$-degree and so are $\field[x]$-linearly independent.
  Also $|\mathcal{B}| = \delta - d_0$, so if $\mathcal{B}$ generates $I_\delta$ then it is a basis of it and the rank of $I_\delta$ is $\delta - d_0$.
  It remains to show that $\mathcal{B}$ generates $I_\delta$, so take some $f \in I_\delta$.
  Since $f \in I$ the multivariate division algorithm using $G$ and the order $\ordLex$ results in $q_0,\ldots,q_s \in \field[x,y]$ such that $f = q_0 b_0 + \ldots + q_s b_s$ with $\ydeg q_i \leq \ydeg f - \ydeg b_i$.
  Since $\ydeg f < \delta$ this means $q_{\hat s+1} = \ldots = q_s = 0$.
  Say that in each iteration of the division algorithm, we use the greatest index $i$ for which $\LTlex(b_i)$ divides the leading term of the current remainder.
  Thus no term of $q_i b_i$ is divisible by $\LTlex(b_{i+1})$ for any $i < s$.
  But by \cref{cor:lazard} then $\LCy(b_{i+1})$ divides $\LCy(b_i)$, and so if $\ydeg(q_i b_i) \geq \ydeg b_{i+1}$ then $\LTlex(b_{i+1}) \mid \LTlex(q_i b_i)$.
  Consequently $\ydeg q_i < \ydeg b_{i+1} - \ydeg b_i$, and therefore $f$ is a $\field[x]$-linear combination of the elements of $\mathcal{B}$.
\end{proof}

\begin{lemma}
  \label{lem:fast_rem}
  There is an algorithm which inputs a \(\ordLex\)-\Grobner basis $G = [ b_0,\ldots,b_s] \subseteq \field[x,y]$ with $\ydeg b_0 = 0$, and a polynomial $f \in \field[x,y]$, and outputs $f \rem G$ in time $\softO{|G|d_x(\ydeg f)}$, where $d_x = \max(\xdeg f, \xdeg b_0)$.
\end{lemma}
\begin{proof}
  \def\xbardeg{\deg^\circ_x}
  \def\ybardeg{\deg^\circ_y}
  This is a special case of \cite{van_der_hoeven_complexity_2015}: the multivariate division algorithm computes $q_0,\ldots,q_s,R \in \field[x,y]$ such that $f = q_0 b_0 + \ldots + q_s b_s + R$ with $R = f \rem G$, and the cost of the algorithm can be bounded as
  \[
    \textstyle
    \sum_{i=0}^s \xbardeg(q_i b_i)\ybardeg(q_i b_i) + \xbardeg(R)\ybardeg(R) \ ,
  \]
  where $\xbardeg(\cdot)$ denotes an a priori upper bound on the $x$-degree, and similarly for $\ybardeg(\cdot)$.
  Since $G$ is a \(\ordLex\)-\Grobner basis, then $\ybardeg(q_i b_i) \leq \ydeg f$ and $\ybardeg(R) \leq \ydeg f$.
  For the $x$-degrees, note that in an iteration of the division algorithm where $b_i, i > 0$ is used, then $\xdeg \tilde R < \xdeg b_0$, where $\tilde R$ is the current remainder, since otherwise the algorithm would have reduced by $b_0$ as $\ydeg b_0 = 0$.
  Hence $\xdeg(q_i) \leq \xdeg(q_i \LTlex(b_i)) < \xdeg b_0$ and so $\xbardeg(q_i b_i) \leq 2\xdeg b_0$.
  Similarly, $\xbardeg(R) < \xdeg b_0$.
  Left is only $\xbardeg(q_0b_0)$: since $q_0 b_0 = f - q_1 b_1 - \ldots - q_s b_s - R$, then $\xdeg(q_0 b_0) \leq \max(\xdeg f,\ 2\xdeg b_0)$.
\end{proof}

\end{document}